\documentclass[aip,preprint,showkeys]{revtex4-2}

\usepackage{amsmath,amsthm,amssymb,mathrsfs,mhchem,tensor}
\usepackage{extarrows}
\usepackage{lineno}  

\theoremstyle{thmstyleone}%
\newtheorem{theorem}{Theorem}
\newtheorem{proposition}[theorem]{Proposition}%

\theoremstyle{thmstyletwo}%
\newtheorem{remark}{Remark}%

\theoremstyle{thmstylethree}%
\newtheorem{definition}{Definition}%
\newtheorem*{principle}{Quantum Bayes' Rule}
\begin{document}

\title{A Quantum Bayes' Rule and Related Inference}

\author{Huayu Liu}  

\affiliation{%
 School of Statistics and Mathematics, Central University of Finance and Economics, South Road, Beijing, 100081, China\\
}%

\date{\today}

\begin{abstract}
In this work a quantum analogue of Bayesian inference is considered. Based on the notion of instrument, we propose a quantum analogue of Bayes' rule, which elaborates how a prior normal state updates under observations. Besides, we investigate the limit of posterior normal state as the number of observations goes to infinity. After that, we generalize the fundamental notions and results of Bayesian inference according to quantum Bayes' rule. It is noted that our theory not only retains the classical one as a special case but possesses many new features as well.
\end{abstract}

\maketitle

\section{Introduction}\label{sec1}
Bayesian inference was born in the 1920s and reached its heyday in the 1950s, and is now widely used in science and engineering. In the past century, quantum physics has flourished and merged with many fields.
About fifty years ago, an interdisciplinary field called quantum statistics 
was born \cite{bib1,bib2}. In this field there are many studies in which Bayesian analysis are applied to quantum physics, as shown in \cite{bib3,bib4,bib5,bib6,bib7,bib8}.
What is at the heart of Bayesian analysis is Bayes' rule, which elaborates how a prior distribution updates under observations. Intuitively, the update of a prior distribution according to Bayes' rule is quite similar to the update of a quantum state according to a specific measurement. We will see that the former is a special case of the latter. To understand this, first we have to embed a probability space into a complex Hilbert space. A probability space is a triad $(X,\mathscr{A},{\rm P})$ where $(X,\mathscr{A})$ is a measurable space and ${\rm P}$ a probability measures on $(X,\mathscr{A})$. For simplicity, here we consider a probability space with finite samples.

Let $(X,2^{X},{\rm P})$ be a probability space where $X=\{x_{1},\cdots,x_{n}\}$, $2^{X}$ the power of $X$ (i.e. the set of all subsets of $X$) and ${\rm P}$ a probability measure on $(X,2^{X})$. Let $\mathbb{H}$ be a $n$-dimensional complex Hilbert space and $\{\psi_{j}\}_{j=1}^{n}$ an orthonormal basis of $\mathbb{H}$. Denote by $c$ the one-to-one correspondence between $\{x_{i}\}$ and the projector $|\psi_{i}\rangle\langle\psi_{i}|$ for all $i\in[n]$, where $[n]:=\{1,2,\cdots,n\}$. Then there is a density operator $\rho=\sum_{i=1}^{n}{\rm P}(\{x_{i}\})|\psi_{i}\rangle\langle\psi_{i}|$ such that
\begin{align}
{\rm P}(A)&=\sum_{j\in\{i:x_{i}\in A\}}\text{tr}(\rho|\psi_{j}\rangle\langle\psi_{j}|)\\
&=\text{tr}(\rho P_{A})
\end{align}
for all $A\in 2^{X}$, where $P_{A}=\sum_{j\in\{i: x_{i}\in A\}}|\psi_{j}\rangle\langle\psi_{j}|$ is a projector.
Therefore, the conditional probability
\begin{align}
    {\rm P}(B|A)&=\cfrac{{\rm P}(B\cap A)}{{\rm P}(A)}\\
    &=\cfrac{\text{tr}(\rho P_{B}P_{A})}{\text{tr}(\rho P_{A})}\\
    &=\cfrac{\text{tr}(P_{B}P_{A}\rho P_{A}P_{B})}{\text{tr}(P_{A}\rho P_{A})}\\
    &=\text{tr}(P_{B}\rho_{A}P_{B})
\end{align}
for all $A,B\in2^{X}$ such that ${\rm P}(A)\neq 0$, where $\rho_{A}=P_{A}\rho P_{A}/\text{tr}(P_{A}\rho P_{A})$ is a density operator. Readers familiar with
quantum measurement may have recognized that operations $\rho\mapsto P_{A}\rho P_{A},\ \forall A\in 2^{X}$ indeed defines a L$\mathrm{\ddot{u}}$ders instrument, which describes the statistical properties of a specific family of measurements. The definition of instrument will be introduced in Sect.\ref{sec2}.

With the ingredients above, we now turn to the connection between Bayes' rule and instrument. Assume that the unknown parameter $\theta$ takes values in the set $\Theta$ that consists of $m$ elements and has a prior distribution $\Pi$. Assume that the observation $x$ takes values in the set $X$ that consists of $n$ elements and has the conditional distribution ${\rm P}_{\theta}$. Let $\mathbb{H}_{0}$ be a $m$-dimensional complex Hilbert space with $\{\phi_{l}\}_{l=1}^{m}$ an orthonormal basis and $\mathbb{H}_{1}$ a $n$-dimensional complex Hilbert space with $\{\psi_{j}\}_{j=1}^{n}$ an orthonormal basis. Denote by $c_{0}$ (resp. $c_{1}$) the one-to-one correspondence between $\{\theta_{k}\}$ and the projector $|\phi_{k}\rangle\langle\phi_{k}|$ for all $k\in[m]$ (resp. $\{x_{i}\}$ and $|\psi_{i}\rangle\langle\psi_{i}|$ for all $i\in[n]$). In order to understand the connection between Bayes' rule and instrument, first we have to embed the product probability space $(\Theta\times X,2^{\Theta}\times 2^{X},\Pi\times{\rm P}_{\theta})$ into a finite-dimensional complex Hilbert space, where $\Theta\times X$ is the Cartesian product, $2^{\Theta}\times 2^{X}$ the product $\sigma$-algebra, $\Pi\times {\rm P}_{\theta}$ the product probability measure on the measurable space $(\Theta\times X,2^{\Theta}\times 2^{X})$. But for brevity, we denote $\Pi\times{\rm P}_{\theta}$ as ${\rm P}_{\times}$ in the following. A convenient choice for the finite-dimensional complex Hilbert space for embedding is the tensor product $\mathbb{H}_{0}\otimes\mathbb{H}_{1}$ since $\{\phi_{l}\otimes\psi_{j}:l\in[m],j\in[n]\}$ is an orthonormal basis of $\mathbb{H}_{0}\otimes\mathbb{H}_{1}$ so that there is a natural one-to-one correspondence $c_{01}$ between $\{(\theta_{k},x_{i})\}$ and the projector $|\phi_{k}\otimes\psi_{i}\rangle\langle\phi_{k}\otimes\psi_{i}|$ for all $k\in[m],i\in[n]$. As before, there is a density operator $\rho=\sum_{k,i}{\rm P}_{\times}(\{(\theta_{k},x_{i})\})|\phi_{k}\otimes\psi_{i}\rangle\langle\phi_{k}\otimes\psi_{i}|$ such that
\begin{align}
{\rm P}_{\times}(E)&=\sum_{(l,j)\in\{(k,i):(\theta_{k},x_{i})\in E\}}\text{tr}(\rho|\phi_{l}\otimes\psi_{j}\rangle\langle\phi_{l}\otimes\psi_{j}|)\\
&=\text{tr}(\rho P_{E})
\end{align}
for all $E\in 2^{\Theta}\times 2^{X}$, where $P_{E}=\sum_{(l,j)\in\{(k,i):(\theta_{k},x_{i})\in E\}}|\phi_{l}\otimes\psi_{j}\rangle\langle\phi_{l}\otimes\psi_{j}|$ is a projector. Then the posterior density of $\theta$ under the observation $x$ (given ${\rm P}_{\times}(\{(\cdot,x)\})>0$) is
\begin{align}
    \pi(\theta|x)&=\cfrac{{\rm P}_{\times}(\{(\theta,x)\})}{{\rm P}_{\times}(\{(\cdot,x)\})}\\
    &=\text{tr}(P_{\{(\theta,\cdot)\}}\rho_{\{(\cdot,x)\}}P_{\{(\theta,\cdot)\}})
\end{align}
for all $\theta\in\Theta$, where $\{(\cdot,x)\}:=\{(\theta_{k},x):k\in[m]\}$
and $\rho_{\{(\cdot,x)\}}=P_{\{(\cdot,x)\}}\rho P_{\{(\cdot,x)\}}/\text{tr}(P_{\{(\cdot,x)\}}\\
\rho P_{\{(\cdot,x)\}})$ is a density operator. Again operations $\rho\mapsto P_{\{(\cdot,x)\}}\rho P_{\{(\cdot,x)\}}$ defines a L$\mathrm{\ddot{u}}$ders instrument.

The analysis above suggests that the Bayesian update  $\pi(\theta)\mapsto\pi(\theta|x)$ can be realized by first carrying out the instrumental update $\rho\mapsto\rho_{\{(\cdot,x)\}}$ and then applying Born's rule. In other words, through embedding, the update of a prior distribution according to Bayes' rule can be identified with the update of a prior normal state (describing a quantum system) according to a L$\mathrm{\ddot u}$ders instrument (describing a specific family of measurements).
We assert that this is also true for the general version of Bayesian update
and will give a proof in Sect.\ref{sec3}.

Quantum theory has been used outside of physics for quite some time. In cognition, decision-making and even economics, quantum theory excels at modeling non-commutative phenomena. Recently, \cite{bib33} successfully modeled both Question Order Effect and Response Replicability Effect in human cognition with a quantum system, suggesting that human cognitive processes are probably similar to the evolution of quantum states. In classical Bayesian inference, a Bayesian's knowledge of an object is modeled by a classical system and Bayes' rule tells a Bayesian how to update her knowledge of an object based on observations. Since there is evidence that quantum system models human cognition better, we wonder if there is a quantum analogue of Bayes' rule together with a quantum analogue of Bayesian inference, in which a Bayesian's knowledge of an object is modeled by a quantum system.

To the best of our knowledge, although a few attempts have been made to quantize Bayes' rule \cite{bib9,bib10,bib12,bib13,bib14,bib15}, there is still no research addressing a quantum analogue of Bayesian inference. In the following we will briefly review these articles.

In \cite{bib9}, a quantum analogue of Bayes' rule is proposed and the exact condition with respect to the validity of it is explored.
Bayes' rule is generalized in \cite{bib10} with the prior being a density matrix and the likelihood being a covariance matrix. In \cite{bib12}, an inherently diagrammatic formulation of quantum Bayes' rule is proposed and a necessary and sufficient condition for the existence of Bayesian inverse in the setting of finite-dimensional $C^{*}$-algebras is provided. A quantum analogue of Bayes' rule is put forward in \cite{bib13} based on the notion of operator valued measure and quantum random variable. 
The graphical framework for Bayesian inference raised in \cite{bib14} is sufficiently general to cover both the standard case and the proposals for quantum Bayesian inference in which the degrees of belief are considered to be represented by density operators instead of probability distributions.
The approach of maximizing quantum relative entropy is followed in \cite{bib15} to study a quantum analogue of Bayes' rule, resulting in some generalizations.

In this work we focus on a quantum analogue of Bayesian inference. We find that through proper embedding, the Bayesian update of a prior distribution can be identified with the instrumental update of a prior normal state. Based on the notions of von Neumann algebra, normal state, observable, and instrument, we put forward a quantum analogue of Bayes' rule, which is concise and has clear physical meaning. Besides, as with Bayes' rule, it is compatible with sequential measurement schemes. By the way, we obtain a sufficient condition for sequential measurements to be joint measurements. Parallel to the asymptotic normality of posterior distribution, we obtain two sufficient conditions for the convergence of posterior normal state, and give the definition of the weak consistency of posterior normal state by analogy to the definition of the weak consistency of posterior distribution, and thus obtain two sufficient conditions for the weak consistency of posterior normal state. Then we generalize the fundamental notions and results of Bayesian inference according to quantum Bayes' rule. Fortunately, our theory retains the classical one as a special case, although we note that for a given quantum Bayesian decision problem, a quantum Bayes solution and a quantum posterior solution are generally no longer equivalent.

Our manuscript is organized as follows. 
Sect.\ref{sec2} is devoted to presenting some facts about operator valued measure, instrument and a family of posterior normal states.
Then in Sect.\ref{sec3} and Sect.\ref{sec4}, we focus on a quantum Bayes' rule and the limit of posterior normal state, respectively.
Next in Sect.\ref{sec5}, we move to a quantum analogue of Bayesian inference. Finally Sect.\ref{sec6} is a discussion.

\section{Preliminaries}\label{sec2}

In general the sample space may have infinitely many elements. In order to understand the connection between Bayes' rule and instrument in this case, first we have to embed a probability space with an arbitrary measurable space into a complex Hilbert space. It is not difficult to find that the key to achieving this is to construct a map from the $\sigma$-algebra $\mathscr{A}$ to the set of projectors  $\{P_{A}:A\in\mathscr{A}\}$ so that there is a one-to-one correspondence between the probability measure ${\rm P}$ and the density operator $\rho$. In fact this map is an observable. Observable is one of the important notions in quantum measurement that we will introduce below. Let $(X,\mathscr{A})$ be a measurable space and $\mathbb{H}$ a complex Hilbert space. Denote by $\mathcal{B}^{+}(\mathbb{H})$ the set of positive bounded linear operators on $\mathbb{H}$ and by $\mathcal{G}(\mathbb{H})$ the set of density operators on $\mathbb{H}$.
\begin{definition}
    A map $\nu:\mathscr{A}\to\mathcal{B}^{+}(\mathbb{H})$ is called an observable iff\\
     (\romannumeral1) $\nu(X)=1$;\\
     (\romannumeral2) For any countable collection of sets $\{A_{k}\}_{k\in\mathbb{N}_{+}}\subseteq
    \mathscr{A}$ with $A_{i}\cap A_{j}=\varnothing$ for $i\neq j$ we have
    \begin{align}
        \nu\left(\bigcup_{k=1}^{\infty}A_{k}\right)=\sum_{k=1}^{\infty}\nu(A_{k}),\ \ \text{weakly.}
    \end{align}
\end{definition}
    Moreover, $\nu$ is called a sharp observable iff $\nu^{*}(A)=\nu(A)=\nu^{2}(A)$ for all $A\in\mathscr{A}$.
    The convex set of observables $\nu:\mathscr{A}\to\mathcal{B}^{+}(\mathbb{H})$ is denoted by $\mathcal{O}(\mathscr{A},\mathbb{H})$. 
    A convex combination of $\nu_{1}$ and $\nu_{2}$ in $\mathcal{O}(\mathscr{A},\mathbb{H})$ (i.e. $t\nu_{1}+(1-t)\nu_{2},\ 0<t<1$) can be viewed as a randomization of measuring processes described by $\nu_{1}$ and $\nu_{2}$.
    
    Let $\nu:\mathscr{A}\to\mathcal{B}^{+}(\mathbb{H})$ be an observable. For any density operator $\rho\in\mathcal{G}(\mathbb{H})$, the probability measure $\nu_{\rho}$ induced by $\nu$ is defined by
\begin{align}
    \nu_{\rho}(A) = \text{tr}[\rho\nu(A)],\ \forall A\in\mathscr{A}\text{.}
\end{align}

One may question the capacity of a complex Hilbert space $\mathbb{H}$. In other words, one may wonder whether there is a probability measure ${\rm P}$ on the measurable space $(X,\mathscr{A})$ such that for all observables $\nu\in\mathcal{O}(\mathscr{A},\mathbb{H})$ and all density operators $\rho\in\mathcal{G}(\mathbb{H})$ we have ${\rm P}(A)\neq\text{tr}[\rho\nu(A)]$, for some $A\in\mathscr{A}$.
Fortunately this will never happen. It is shown in \cite{bib27} that for each probability measure ${\rm P}$ on the measurable space $(X,\mathscr{A})$ and each density operator $\rho\in\mathcal{G}(\mathbb{H})$ there is a unique observable $\nu:\mathscr{A}\to\mathcal{B}^{+}(\mathbb{H})$ such that ${\rm P}(A)=\text{tr}[\rho\nu(A)],\ \forall A\in\mathscr{A}$. This guarantees that we can embed any probability space into a complex Hilbert space.

As we discussed in Sect.\ref{sec1}, there is a connection between Bayes' rule and instrument. The notion of instrument is a bit more complicated. To introduce this, first we have to agree on some notations.
Let $\mathscr{M}\subseteq\mathcal{B}(\mathbb{H})$ be a von Neumann algebra.
Denote by $\mathscr{M}_{*}$ the predual of $\mathscr{M}$ (i.e. the set of $\sigma$-weakly continuous bounded linear functionals on $\mathscr{M}$), by $\mathfrak{S}(\mathscr{M})$ the set of normal states on $\mathscr{M}$ (i.e. the set of positive and unit elements of $\mathscr{M}_{*}$), by $\langle\cdot,\cdot\rangle$ the duality pairing between $\mathscr{M}_{*}$ and $\mathscr{M}$ and by $\mathcal{B}^{+}(\mathscr{M}_{*})$ the set of positive bounded linear maps on $\mathscr{M}_{*}$.
    
    A map $\Psi\in\mathcal{B}^{+}(\mathscr{M}_{*})$ is called a subtransition iff it satisfies
    $\langle\Psi\varphi,1\rangle\leq\langle\varphi,1\rangle$ 
    for all $0\leq\varphi\in\mathscr{M}_{*}$. Moreover, if $\Psi$ satisfies 
    $\langle\Psi\varphi,1\rangle=\langle\varphi,1\rangle$  for all $\varphi\in\mathscr{M}_{*}$ 
    then $\Psi$ is called a transition. The dual of a subtransition (resp. transition) $\Psi$, denoted $\Phi$, is defined by $\langle\varphi,\Phi a\rangle=\langle\Psi\varphi,a\rangle$
for all $\varphi\in\mathscr{M}_{*}$ and $a\in\mathscr{M}$.
It is a normal positive linear map on $\mathscr{M}$ such that $\Phi 1\leq 1$ (resp. $\Phi 1=1$).
If $\Psi$ is completely positive (CP), then $\Psi$ is called an operation (resp. channel).
\begin{definition}
    A map $\mathfrak{I}:\mathscr{A}\to\mathcal{B}^{+}(\mathscr{M}_{*})$ is called an instrument 
    (normalized subtransition valued measure) iff
    
    (\romannumeral1) $\mathfrak{I}(X)$ is a transition;
    
    (\romannumeral2) For any countable collection of sets $\{A_{k}\}_{k\in\mathbb{N}_{+}}\subseteq
    \mathscr{A}$ with $A_{i}\cap A_{j}=\varnothing$ for $i\neq j$ we have 
    \begin{align}\label{Equ10}
        \mathfrak{I}\left(\bigcup_{k=1}^{\infty}A_{k}\right)=\sum_{k=1}^{\infty}\mathfrak{I}(A_{k})\text{,}
    \end{align}
    where the convergence on the right side of the equation $(\ref{Equ10})$ is with respect to the strong operator topology of $\mathcal{B}^{+}(\mathscr{M}_{*})$.
\end{definition}
The dual of an instrument $\mathfrak{I}$, denoted $\mathfrak{I}^{*}$, is defined by 
$\langle\varphi,\mathfrak{I}^{*}(A)a\rangle=\langle\mathfrak{I}(A)\varphi,a\rangle$
for all $A\in\mathscr{A},\ \varphi\in\mathscr{M}_{*}$ and $a\in\mathscr{M}$. 
It is a map from $\mathscr{A}$ to $\mathcal{L}^{+}(\mathscr{M})$ (i.e. the set of positive linear maps on $\mathscr{M}$).
The convex set of instruments $\mathfrak{I}:\mathscr{A}\to\mathcal{B}^{+}(\mathscr{M}_{*})$ is denoted by 
${\rm Ins}(\mathscr{A},\mathscr{M}_{*})$.
If $\mathfrak{I}(A)$ is an operation for all $A\in\mathscr{A}$, then 
$\mathfrak{I}$ is called a CP instrument.
The convex set of CP instruments $\mathfrak{I}:\mathscr{A}\to\mathcal{B}^{+}(\mathscr{M}_{*})$ is denoted by ${\rm CPIns}(\mathscr{A},\mathscr{M}_{*})$.

An instrument naturally induces an observable. Let $\mathfrak{I}:\mathscr{A}\to\mathcal{B}^{+}(\mathscr{M}_{*})$ be an instrument. The observable $\nu$ induced by $\mathfrak{I}$ is defined by
\begin{align}
\nu(A)=\mathfrak{I}^{*}(A)1,\ \forall A\in\mathscr{A}\text{.}
\end{align} 
Then for any $\varphi\in\mathfrak{S}(\mathscr{M})$, 
the probability measure ${\rm P}_{\varphi}$ induced by $\nu$ is defined by
\begin{align}
    {\rm P}_{\varphi}(A)=\langle\varphi,\nu(A)\rangle,\ \forall A\in\mathscr{A}.
\end{align}

Note that different instruments may induce the same observable. It is shown in \cite{bib25} that each instrument induces a unique observable and each observable can be induced by a unique class of instruments.

To measure a quantum system one needs an apparatus. In standard experimental scenarios, an apparatus will reveal the following two facts to the surveyor, called the statistical properties of an apparatus: (\romannumeral1) Probability that each measurement outcome occurs given a quantum state; (\romannumeral2) Update of quantum states given each measurement outcome. However, different apparatuses may have the same statistical properties so that it would be a little redundant to distinguish them if we only care about their statistical properties. Hence the notion of instrument was put forward. It is shown in \cite{bib18} that each apparatus corresponds to a unique instrument and each instrument corresponds to a unique statistical equivalence class of apparatuses.

Let $\mathfrak{I}:\mathscr{A}\to\mathcal{B}^{+}(\mathscr{M}_{*})$ be an instrument. For any $A\in\mathscr{A}$, ${\rm P}_{\varphi}(A)$ is 
the probability that $A$ occurs (iff the measurement outcome $x\in A$) when a quantum system $S$ described by a von Neumann algebra $\mathscr{M}$ is in the normal state $\varphi\in\mathfrak{S}(\mathscr{M})$ and a measurement with an apparatus corresponding to $\mathfrak{I}$ is performed. Furthermore, $\varphi_{A}=\mathfrak{I}(A)\varphi/{\rm P}_{\varphi}(A)$ is the normal state that the quantum system $S$ is in immediately after $A$ occurs (provided that ${\rm P}_{\varphi}(A)>0$).

Unfortunately, not all apparatuses are physically realizable. It is shown in \cite{bib27} that the corresponding instrument of a physically realizable apparatus should be CP. Furthermore, there is a one-to-one correspondence between the statistical equivalence classes of physically realizable apparatuses and CP instruments. In other words, each CP instrument can be realized physically by at least one apparatus.

It is known that the measurement outcome $\{x\}$ is an elementary event. However, it may happen that ${\rm P}(\{x\})=0$. So how will quantum states update given a measurement outcome $\{x\}$ such that ${\rm P}(\{x\})=0$? This question is settled by the notion of a family of posterior normal states. Let $\mathbb{H}$ be a complex Hilbert space, $\mathscr{M}\subseteq\mathcal{B}(\mathbb{H})$ a von Neumann algebra, $\mathfrak{I}:\mathscr{A}\to\mathcal{B}^{+}(\mathscr{M}_{*})$ an instrument and $\varphi\in\mathfrak{S}(\mathscr{M})$ a normal state.
\begin{definition}
    The set $\{\varphi_{x}:x\in X\}$ is 
    called a family of posterior normal states w.r.t. $(\mathfrak{I},\varphi)$
    iff\\
    (\romannumeral1) $\varphi_{x}$ is a normal state for all $x\in X$;\\
    (\romannumeral2) $\{\varphi_{x}:x\in X\}$ is weakly* ${\rm P}_{\varphi}$ measurable
    (i.e. the function $x\mapsto\langle\varphi_{x},a\rangle$ is ${\rm P}_{\varphi}$ measurable for any $a\in\mathscr{M}$);\\
    (\romannumeral3) For any $a\in\mathscr{M}$ and $A\in\mathscr{A}$,
    \begin{align}\label{Equ21}
    \int_{A}\langle\varphi_{x},a\rangle{\rm P}_{\varphi}(dx)=\langle\mathfrak{I}(A)\varphi,a\rangle\text{.}
    \end{align}
\end{definition}
If $\{\varphi_{x}:x\in X\}$ is a family of posterior normal states w.r.t. $(\mathfrak{I},\varphi)$, then it exists uniquely in the sense that if $\{\varphi_{x}':x\in X\}$ is another family of posterior normal states w.r.t. $(\mathfrak{I},\varphi)$ then for any $a\in\mathscr{M}$, $\langle\varphi_{x},a\rangle=\langle\varphi_{x}',a\rangle$ ${\rm P}_{\varphi}$-a.s.. The existence of a family of posterior normal states w.r.t. $(\mathfrak{I},\varphi)$ is discussed thoroughly in \cite{bib20}.

\section{Quantum Bayes' Rule}\label{sec3}
To illustrate our theory, and to answer the question remained in Sect.\ref{sec1}, we begin this section by exploring the connection between the general version of Bayes' rule and instrument.
Recall that in Sect.\ref{sec1} the product probability space $(\Theta\times X,2^{\Theta}\times2^{X},{\rm P}_{\times})$ is embedded into a $(m\times n)$-dimensional complex Hilbert space $\mathbb{H}$ through a sharp observable $B\mapsto P_{B}$, as a result of which, the update of a prior distribution according to Bayes' rule can be identified with the update of a prior density operator according to a L$\mathrm{\ddot{u}}$ders instrument. However, it suffices to reach the same conclusion by embedding the probability spaces $(\Theta,2^{\Theta},\Pi)$ and $(X,2^{X},\int1_{(\cdot)}{\rm P}_{\theta}d\Pi)$ into an arbitrary complex Hilbert space whose dimension is no less than $m\times n$ through a sharp observable $\nu:2^{\Theta}\to\mathcal{B}^{+}(\mathbb{H})$ and a L$\mathrm{\ddot{u}}$ders instrument $\mathfrak{I}:2^{X}\to\mathcal{B}^{+}(\mathcal{B}_{1}(\mathbb{H}))$ inducing a sharp observable $\lambda$ compatible with $\nu$ (i.e. the commutator $[\nu(E),\lambda(A)]=0$ for all $E\in2^{\Theta}$ and $A\in2^{X}$), respectively. Following the idea above we will prove that through proper embedding the update of a prior distribution according to the general version of Bayes' rule is nothing but the update of a prior normal state according to a specific instrument.

    Let $(\Theta,\mathscr{E},\nu)$ and $(X,\mathscr{A},\mu)$
    be $\sigma$-finite measure spaces and $p(x\lvert\theta)$ a nonnegative $\mathscr{E}\times\mathscr{A}$ 
    measurable real valued function satisfying
    \begin{align}
    \int_{X}p(x\lvert\theta)\mu(dx)=1,\ \forall\theta\in\Theta,
    \end{align}
    so that the function
    \begin{align}
    K:(\theta,A)\mapsto\int_{A}p(x\lvert\theta)\mu(dx),\ \forall (\theta,A)\in\Theta\times\mathscr{A}
    \end{align}
    is a Markov kernel. Denote by $L^{\infty}(\nu)\subseteq\mathcal{B}(L^{2}(\nu))$ the abelian von Neumann algebra of $\nu$-a.e. equivalence classes of essentially bounded $\mathscr{E}$ measurable functions from $\Theta$ to the complex field $\mathbb{C}$, where $L^{2}(\nu)$ is the complex Hilbert space of $\nu$-a.e. equivalence classes of square integrable $\mathscr{E}$ measurable functions from $\Theta$ to $\mathbb{C}$. Identify $L^{\infty}(\nu)_{*}$ with $L^{1}(\nu)$ and denote by $L_{+,1}^{1}(\nu)$ the set of normal states. 
    Let $\tau:\mathscr{E}\to L^{\infty}(\nu)$ be a map defined by 
    \begin{align}
        E\mapsto 1_{E},\ \forall E\in\mathscr{E},
    \end{align}
    where $1_{E}$ is an indicator. In fact $\tau$ is an observable since (\romannumeral1) $1_{E}\ge 0$ for all $E\in\mathscr{E}$; (\romannumeral2) $1_{\Theta}$ is the identity and (\romannumeral3) $1_{\cup_{k\ge 1} E_{k}}=\sum_{k\ge 1}1_{E_{k}}$ for any countable collection of mutually disjoint measurable sets $\{E_{k}\}_{k\ge 1}$.
    Therefore, we can embed any probability space $(\Theta,\mathscr{E},\Pi)$ satisfying $\Pi\ll\nu$ (i.e. $\Pi$ is absolutely continuous with respect to $\nu$) into the complex Hilbert space $L^{2}(\nu)$ through the observable $\tau$. Note that in the case mentioned in Sect.\ref{sec1} we can not determine the normal state (density operator) uniquely if we only embed $(\Theta,2^{\Theta},\Pi)$ into $\mathbb{H}$ through a sharp observable (even after embedding both if the dimension of $\mathbb{H}$ is larger than $m\times n$). However, the normal state can be uniquely determined if we embed $(\Theta,\mathscr{E},\Pi)$ into $L^{2}(\nu)$ through $\tau$ since by Radon-Nikod${\rm \acute{y}}$m theorem there is a one-to-one correspondence between probability measures that are absolutely continuous with respect to $\nu$ and normal states on $L^{\infty}(\nu)$. 
    
    Let ${\rm P}^{\circ}$ be a probability measure on $(X,\mathscr{A})$ defined by ${\rm P}^{\circ}(A)=\int_{\Theta}1(\theta)K(\theta,A)\pi(\theta)\nu(d\theta)$ for all $A\in\mathscr{A}$, where $\pi(\theta)\in L_{+,1}^{1}(\nu)$ (i.e. $\pi(\theta)\ge 0$ and $\int\lvert\pi\rvert d\nu=1$) is a normal state. We have to embed $(X,\mathscr{A},{\rm P}^{\circ})$ into $L^{2}(\nu)$ through an instrument. 
    Let $\mathfrak{I}:\mathscr{A}\to\mathcal{B}^{+}(L^{1}(\nu))$ be a map defined by
    \begin{align}
       A\mapsto K(\theta,A)\cdot,\ \forall A\in\mathscr{A},\\
        K[f]:=[Kf],\ \forall [f]\in L^{1}(\nu).
    \end{align}
    Needless to say, $\mathfrak{I}$ is an instrument, as verified below.
    (\romannumeral1) $\int_{\Theta}1(\theta)K(\theta,A)f(\theta)\nu(d\theta)\leq\int_{\Theta}1(\theta)f(\theta)\nu(d\theta)$ for all $A\in\mathscr{A}$ and $L^{1}(\nu)\ni f\ge 0$ since $0\leq K(\theta,A)\leq 1$ for all $\theta\in\Theta$ and $A\in\mathscr{A}$.
    (\romannumeral2) $\int_{\Theta}1(\theta)K(\theta,X)f(\theta)\nu(d\theta)=\int_{\Theta}1(\theta)f(\theta)\nu(d\theta)$ for all $A\in\mathscr{A}$ and $L^{1}(\nu)\in f\ge 0$ since $K(\theta,X)=1$ for all $\theta\in\Theta$.
    (\romannumeral3)
    The $\sigma$-additivity of $\mathfrak{I}$ follows directly from the $\sigma$-additivity of $K(\theta,\cdot)$ as a probability measure for all $\theta\in\Theta$. 
    
    The following result shows that $\{\pi(\theta|x):x\in X\}$ is indeed a family of posterior normal states w.r.t. $(\mathfrak{I},\pi(\theta))$, where $\pi(\theta|x)=p(x|\theta)\pi(\theta)/\int_{\Theta}p(x|\theta)\pi(\theta)\nu(d\theta)$ and $\pi(\theta)\in L_{+,1}^{1}(\nu)$.
    \begin{theorem}\label{theo1}
         $\{\pi(\theta|x):x\in X\}$ is a family of posterior normal states w.r.t. $(\mathfrak{I},\pi(\theta))$. Moreover, it is unique in the sense that if $\{\pi(\theta|x)':x\in X\}$ is another family of posterior normal states w.r.t. $(\mathfrak{I},\pi(\theta))$, then $\pi(\theta|x)$ and $\pi(\theta|x)'$ only differ on a $(\nu\times\mu)$-null set. 
    \end{theorem}
\begin{proof}
    Denote by $\pi(\theta)$ an element of $\pi(\theta)$. 
    (\romannumeral1)
    Obviously $p(x|\theta)\pi(\theta)$ is a nonnegative $(\nu\times\mu)$-a.e., $\mathscr{E}\times\mathscr{A}$ measurable real valued function such that $\int_{\Theta}p(x|\theta)\pi(\theta)\nu(d\theta)\ge 0$ for all $x\in X$ and $\int_{\Theta}|\pi(\theta|x)|\nu(d\theta)=1$.
    Thus $\pi(\theta|x)$ is a normal state for all $x\in X$.
    (\romannumeral2)
    Denote by $N_{0}$ the $\nu$-null set such that $\pi(\theta)\ge0$ on $N_{0}^{c}$. 
    We first show that $\int p\pi d(\nu\times\mu)<\infty$.
    \begin{align}
        \int p\pi d(\nu\times\mu)&=\int1_{N^{c}}p\pi d(\nu\times\mu)\\
        &=\int_{\Theta}\pi(\theta)\nu(d\theta)\int_{X}1_{N^{c}}(\theta,x)p(x|\theta)\mu(dx)\ \ \ (\text{by Tonelli theorem})\\
        &=\int_{\Theta} K(\theta,N^{c}_{\theta})\pi(\theta)\nu(d\theta)\\
        &=\int_{\Theta} K(\theta,N^{c}_{\theta})\Pi(d\theta)<\infty,
    \end{align}
    where $N^{c}$ is the complement of the $(\nu\times\mu)$-null set $N=N_{0}\times X$ and $N_{\theta}^{c}$ the section of $N^{c}$ at $\theta$. By Tonelli theorem the function $x\mapsto\int1_{N^{c}}p\pi d\nu=\int p\pi d\nu$ is nonnegative and $\mathscr{A}$ measurable.
    Let $g\in L^{\infty}(\nu)$ and denote by $g$ an element of $g$. Then we have to show that the function $x\mapsto\int gp\pi d\nu$ is ${\rm P}^{\circ}$ measurable. Since $g=\Re(g)+{\rm i}\Im(g)=[\Re(g)]^{+}-[\Re(g)]^{-}+{\rm i}\lbrace[\Im(g)]^{+}-[\Im(g)]^{-}\rbrace$, the integral $\int gp\pi d\nu=\int[\Re(g)]^{+}p\pi d\nu-\int [\Re(g)]^{-}p\pi d\nu+{\rm i}\lbrace\int[\Im(g)]^{+}p\pi d\nu-\int[\Im(g)]^{-}p\pi d\nu\rbrace$, where $\Re(\cdot),\Im(\cdot),(\cdot)^{+},(\cdot)^{-}$ denote the real part, imaginary part, positive part and negative part of a number, respectively. 
    Thus it suffices to show that the function $\ell:x\mapsto\int[\Re(g)]^{+}p\pi d\nu$ is ${\rm P}^{\circ}$ measurable. To do this, we first show that $\int[\Re(g)]^{+}p\pi d(\nu\times\mu)<\infty$.
    \begin{align}
        \int[\Re(g)]^{+}p\pi d(\nu\times\mu)&=\int1_{N^{c}}[\Re(g)]^{+}p\pi d(\nu\times\mu)\\
        &=\int_{\Theta}\Re^{+}[g(\theta)]\pi(\theta)\nu(d\theta)\int_{X}1_{N^{c}}(\theta,x)p(x|\theta)\mu(dx)\\
        &=\int_{\Theta}\Re^{+}[g(\theta)]K(\theta,N^{c}_{\theta})\pi(\theta)\nu(d\theta)\\
        &\leq\int_{\Theta}\lvert\Re(g)\rvert K(\theta,N^{c}_{\theta})\Pi(d\theta)\\
        &\leq\int_{\Theta}\lVert g\rVert_{\infty}K(\theta,N^{c}_{\theta})\Pi(d\theta)<\infty,
    \end{align}
    where $\lVert\cdot\rVert_{\infty}$ is the essential supremum norm on $L^{\infty}(\nu)$. Then by Fubini theorem the nonnegative function $\ell$ is $\mathscr{A}$ measurable but only finite $\mu$-a.e.. Denote by $N_{1}$ the $\mu$-null set such that $\ell$ is finite on $N_{1}^{c}$.
    Note that $1_{N_{1}^{c}}\ell$ is a $\mathscr{A}$ measurable real valued function and equal to $\ell$ $\mu$-a.e.. This implies that $1_{N_{1}^{c}}\ell$ is equal to $\ell$ ${\rm P}^{\circ}$-a.s. since ${\rm P}^{\circ}$ is absolutely continuous with respect to $\mu$. Due to the fact that a function $f$ from $\Theta$ to the extended real line is $\mu$ measurable if and only if there is $\mathscr{A}$ measurable extended real valued function $h$ such that $f=h$ $\mu$-a.e., we have actually proved that $\ell$ is ${\rm P}^{\circ}$ measurable. By the same way, the functions $x\mapsto\int [\Re(g)]^{-}p\pi d\nu,\ x\mapsto\int[\Im(g)]^{+}p\pi d\nu$ and $x\mapsto\int[\Im(g)]^{-}p\pi d\nu$ are all ${\rm P}^{\circ}$ measurable so that the function $x\mapsto\int gp\pi d\nu$ is ${\rm P}^{\circ}$ measurable. Hence the function $x\mapsto\int_{\Theta}g(\theta)\pi(\theta|x)\nu(d\theta)$ is ${\rm P}^{\circ}$ measurable (and also $\mu$ measurable) for all $g\in L^{\infty}(\nu)$.
    (\romannumeral3) Let $g\in L^{\infty}(\nu)$ and denote by $g$ an element of $g$. Without loss of generality, assume that $g$ is nonnegative. By (\romannumeral2), the function $x\mapsto\int_{\Theta}g(\theta)\pi(\theta|x)\nu(d\theta)$ is nonnegative and $\mu$ measurable. Consequently, there is a nonnegative $\mathscr{A}$ measurable function $r$ that is equal to $x\mapsto\int_{\Theta}g(\theta)\pi(\theta|x)\nu(d\theta)$ $\mu$-a.e.. Thus the integral
    \begin{align}
        \int_{A}\left[\int_{\Theta}g(\theta)\pi(\theta|x)\nu(d\theta)\right]\int_{\Theta}1(\theta)K(\theta,dx)\pi(\theta)\nu(d\theta)
        &=\int_{A}r(x)\int_{\Theta}1(\theta)K(\theta,dx)\pi(\theta)\nu(d\theta)\\
        &=\int_{\Theta}\Pi(d\theta)\int_{A}r(x)K(\theta,dx)\\
        &=\int_{\Theta}\Pi(d\theta)\int_{A}r(x)p(x|\theta)\mu(dx)\\
        &=\int_{A}r(x)\mu(dx)\int_{\Theta}p(x|\theta)\Pi(d\theta).
    \end{align}
    And the integral
    \begin{align}
        \int_{\Theta}g(\theta)K(\theta,A)\pi(\theta)\nu(d\theta)
        &=\int_{\Theta}g(\theta)\Pi(d\theta)\int_{A}p(x|\theta)\mu(dx)\\
        &=\int_{A}\mu(dx)\int_{\Theta}g(\theta)p(x|\theta)\Pi(d\theta)\ \ \ \text{(by Tonelli theorem)}.
    \end{align}
    Apparently $r(x)\int_{\Theta}p(x|\theta)\Pi(d\theta)=\int_{\Theta}g(\theta)p(x|\theta)\Pi(d\theta)$ $\mu$-a.e.. Hence for all $A\in\mathscr{A}$ and $g\in L^{\infty}(\nu)$ we have
    \begin{align}
     \int_{A}\left[\int_{\Theta}g(\theta)\pi(\theta|x)\nu(d\theta)\right]\int_{\Theta}1(\theta)K(\theta,dx)\pi(\theta)\nu(d\theta)=
      \int_{\Theta}g(\theta)K(\theta,A)\pi(\theta)\nu(d\theta)\label{3.20}.
    \end{align}
    
    Conversely, assume that a family of posterior normal states $\{\pi(\theta\lvert x):x\in X\}$ w.r.t. $(\mathfrak{I},\pi(\theta))$ exists, where $\pi(\theta)\in L_{+,1}^{1}(\nu)$ is a normal state.
    Then $\pi(\theta|x)$ shall satisfy the following equation 
    \begin{align}
        \int_{A}\left[\int_{\Theta}g(\theta)\pi(\theta\lvert x)\nu(d\theta)\right]\int_{\Theta}1(\theta)K(\theta,dx)\pi(\theta)\nu(d\theta)
       =\int_{\Theta}g(\theta)K(\theta,A)\pi(\theta)\nu(d\theta)
    \end{align} 
    for all $A\in\mathscr{A}$ and $L^{\infty}(\nu)\ni g\ge0$.
    By Tonelli theorem we have
    \begin{align}
       \int_{A}\left[\int_{\Theta}g(\theta)\pi(\theta\lvert x)\nu(d\theta)\right]\mu(dx)\int_{\Theta}p(x\lvert\theta)\pi(\theta)\nu(d\theta)
       =\int_{A}\mu(dx)\int_{\Theta}g(\theta)p(x\lvert\theta)\pi(\theta)\nu(d\theta)
    \end{align}
    for all $A\in\mathscr{A}$ and $L^{\infty}(\nu)\ni g\ge0$,
    which implies that 
    \begin{align}
        \pi(\theta\lvert x)\int_{\Theta}p(x\lvert\theta)\pi(\theta)\nu(d\theta)
        =p(x\lvert\theta)\pi(\theta)
    \end{align}
     holds on a set $M$ satisfying for $\mu$-a.e. $x$ the section $(M^{c})^{x}$ of $M^{c}$ is a $\nu$-null set. Since $\nu$ and $\mu$ are $\sigma$-finite, this is equivalent to $M^{c}$ is a $(\nu\times\mu)$-null set.
\end{proof}
     It is now clear that the Bayesian update $\pi(\theta)\mapsto\pi(\theta|x)$ of a prior density $\pi(\theta)$ is nothing but the update $\pi(\theta)\mapsto\pi(\theta|x)$ of a prior normal state $\pi(\theta)$ according to the instrument $K$.
     This inspired us to propose a quantum analogue of Bayes' rule. Indeed Bayes' rule says if a Bayesian uses an apparatus corresponding to the instrument $K$ to measure a quantum system $S$ described by the von Neumann algebra $L^{\infty}(\nu)$, she will obtain a measurement outcome $\{x\}$ and disturb the state of $S$ according to $\{x\}$. This disturbance is described by an update of the normal state $\varphi\mapsto\varphi_{x}$, where $\{\varphi_{x}:x\in X\}$ is a family of posterior normal states w.r.t. $(K,\pi(\theta))$.
     But what if a Bayesian uses an apparatus corresponding to an arbitrary instrument to measure a quantum system described by an arbitrary von Neumann algebra? To answer this, we put forward the following quantum Bayes' rule by analogy to the classical one. 
    \begin{principle}
        If a quantum system $S$ described by a von Neumann algebra $\mathscr{M}$ is in the prior normal state $\varphi$ before a measurement with an apparatus corresponding to an instrument $\mathfrak{I}:\mathscr{A}\to\mathcal{B}^{+}(\mathscr{M}_{*})$, then after the measurement $S$ will be in the posterior normal state
        $\varphi_{x}$ if the measurement outcome is $\{x\}$, where $\{\varphi_{x}:x\in X\}$ is a family of posterior normal states w.r.t. $(\mathfrak{I},\varphi)$.
    \end{principle}
        Apparently quantum Bayes' rule retains the classical one as a special case. In practice, one may measure a quantum system $S$ described by a von Neumann algebra $\mathscr{M}$ in a sequential measurement scheme, i.e. perform $n$ measurements with $n$ apparatuses (corresponding to $n$ instruments $\mathfrak{I}_{i}:\mathscr{A}_{i}\to\mathcal{B}^{+}(\mathscr{M}_{*}),i\in[n]$) on $S$ one immediately after another. If the quantum system $S$ is in the prior normal state $\varphi$ before these sequentially performed measurements, then by induction the probability that $A_{1},\cdots,A_{n}$ occurs one by one is $\langle\mathfrak{I}_{n}(A_{n})\circ\cdots\circ\mathfrak{I}_{1}(A_{1})\varphi,1\rangle$ and according to quantum Bayes' rule $S$ will be in the posterior normal state $\varphi_{x_{1}\cdots x_{n}}$ if the measurement outcomes are $\{x_{1}\},\cdots,\{x_{n}\}$.
    
    In classical Bayesian inference the probability space $(X,\mathscr{A},{\rm P}_{\theta})$ is usually a product of $n$ probability spaces $(X_{i},\mathscr{A}_{i},{\rm P}_{\theta}^{(i)}),i\in[n]$. This implies that the components of $x\in X$ is independent.
    Let $p(x_{i}|\theta)$ be a nonnegative $\mathscr{E}\times\mathscr{A}_{i}$ measurable real valued function satisfying $\int_{X_{i}} p(x_{i}|\theta)\mu_{i}(dx_{i})=1$ for all $i\in[n]$, where $\mu_{i}$ is a $\sigma$-finite measure on $(X_{i},\mathscr{A}_{i})$. Then the independence of the components of $x$ is equivalent to $p(x|\theta)=\prod_{i=1}^{n}p(x_{i}|\theta)\ \left(\prod_{i=1}^{n}\mu_{i}\right)$-a.e. for all $\theta\in\Theta$. 
    But for convenience, the independence of the components of $x$ is usually guaranteed by $p(x|\theta)=\prod_{i=1}^{n}p(x_{i}|\theta)$ for all $\theta\in\Theta$.
    And since the point-wise multiplication of two real valued functions is commutative, the update of a prior normal state $\pi(\theta)$ according to the instrument $K=\int 1_{(\cdot)}p(x|\theta)\left(\prod_{i=1}^{n}\mu_{i}\right)(dx)$ is equivalent to the sequential update of $\pi(\theta)$ according to $n$ instruments 
    $K_{i}=\int 1_{(\cdot)}p(x_{i}|\theta)\mu_{i}(dx_{i}),i\in[n]$ in any order.
    
    In the case above it holds that $K(\theta,\prod_{i=1}^{n}A_{i})=\prod_{i=1}^{n}K_{i}(\theta,A_{i})$ for all 
    $A_{i}\in\mathscr{A}_{i}$, i.e. $K$ is the unique kernel-like instrument that is equal to the composition of $n$ instruments $K_{i},i\in[n]$ on the measurable rectangles $\prod_{i=1}^{n}A_{i}$.
    In fact under some topological assumptions, the composition of a finite number of instruments can be uniquely extended to an instrument in the following sense. Let $(X_{i},\mathscr{B}(X_{i})),i\in[n]$ be measurable spaces 
    where $X_{i}$ is a second countable locally compact Hausdorff space and $\mathscr{B}(X_{i})$ the Borel $\sigma$-algebra on $X_{i}$. 
    Let $\mathfrak{I}_{i}:\mathscr{B}(X_{i})\to\mathcal{B}^{+}(\mathscr{M}_{*}),i\in[n]$ be instruments. Then according to \cite{bib24}, there is a unique instrument $\mathfrak{I}\in{\rm Ins}(\Pi_{i=1}^{n}\mathscr{B}(X_{i}),\mathscr{M}_{*})$ such that
\begin{align}
        \mathfrak{I}(A_{1}\times\cdots\times A_{n})=\mathfrak{I}_{n}(A_{n})\circ\cdots\circ\mathfrak{I}_{1}(A_{1}),\ \forall A_{i}\in\mathscr{B}(X_{i})\text{.}
\end{align}
 Thus if for all $i\in[n]$ $X_{i}$ is a second countable locally compact Hausdorff space and $\mathscr{A}_{i}$ the Borel $\sigma$-algebra on $X_{i}$, then $K$ is the unique instrument that agrees with the composition of $K_{i},i\in[n]$ on the measurable rectangles $\prod_{i=1}^{n}A_{i}$.

Note that $\varphi_{\prod_{i=1}^{n}A_{i}}=\varphi_{A_{1}\cdots A_{n}}$ provided that ${\rm P}_{\varphi}(\prod_{i=1}^{k}A_{i})>0,k\in[n]$ so it is natural to ask whether a family of posterior normal states w.r.t. $(\mathfrak{I},\varphi)$ is $\{\varphi_{x_{1}\cdots x_{n}}:x_{i}\in X_{i},i\in[n]\}$.
To see this, it suffices to consider the case $n=2$. Let $\mathfrak{I}_{i}:\mathscr{B}(X_{i})\to\mathcal{B}^{+}(\mathscr{M}_{*}),i\in[2]$ be instruments.
\begin{theorem}\label{the2}
     Assume that the map $x_{1}\mapsto\langle\mathfrak{I}_{2}(A_{2})\varphi_{x_{1}},1\rangle$ is $\mathscr{B}(X_{1})$ measurable for all $A_{2}\in\mathscr{B}(X_{2})$.
     Then $\{\varphi_{x_{1}x_{2}}:x_{i}\in X_{i},i\in[2]\}$ is a family of posterior normal states w.r.t. $(\mathfrak{I},\varphi)$.
\end{theorem}
\begin{proof} 
(\romannumeral1) Apparently, $\varphi_{x_{1}x_{2}}$ is a normal state for all $(x_{1},x_{2})\in X_{1}\times X_{2}$. 
(\romannumeral2) Since $\{\varphi_{x_{1}}:x_{1}\in X_{1}\}$ is a family of posterior normal states w.r.t. $(\mathfrak{I}_{1},\varphi)$, we have $\int_{A_{1}}\langle\varphi_{x_{1}},a\rangle{\rm P}_{\varphi}(dx_{1})=\langle\mathfrak{I}_{1}(A_{1})\varphi,a\rangle$ for all $a\in\mathscr{M}$ and $A_{1}\in\mathscr{B}(X_{1})$. Setting $a=\mathfrak{I}_{2}^{*}(A_{2})b,\ 0\leq b\in\mathscr{M}$ and noticing that $\{\varphi_{x_{1}x_{2}}:x_{2}\in X_{2}\}$ is a family of posterior normal states w.r.t. $(\mathfrak{I}_{2},\varphi_{x_{1}})$, we have 
\begin{align}
    \int_{A_{1}}{\rm P}_{\varphi}(dx_{1})\int_{A_{2}}\langle\varphi_{x_{1}x_{2}},b\rangle{\rm P}_{\varphi_{x_{1}}}(dx_{2})
    &=\langle\mathfrak{I}_{2}(A_{2})\circ\mathfrak{I}_{1}(A_{1})\varphi,b\rangle\\
    &=\int_{A_{1}\times A_{2}}\langle\varphi_{x},b\rangle\langle\mathfrak{I}(dx)\varphi,1\rangle,
\end{align}
for all $A_{1}\in\mathscr{B}(X_{1})$ and $A_{2}\in\mathscr{B}(X_{2})$. Since the map $x_{1}\mapsto\langle\mathfrak{I}_{2}(A_{2})\varphi_{x_{1}},1\rangle$ is $\mathscr{B}(X_{1})$ measurable (of course ${\rm P}_{\varphi}$ measurable) for all $A_{2}\in\mathscr{B}(X_{2})$, $\langle\mathfrak{I}_{2}(A_{2})\varphi_{x_{1}},1\rangle$ is a Markov kernel.
As a result, ${\rm P}_{\varphi}\times{\rm P}_{\varphi_{x_{1}}}$ is the unique probability measure on the measurable space $(X_{1}\times X_{2},\mathscr{B}(X_{1}\times X_{2}))$ such that 
\begin{align}
{\rm P}_{\varphi}\times{\rm P}_{\varphi_{x_{1}}}(A_{1}\times A_{2})
&=\int_{A_{1}}\langle\mathfrak{I}_{2}(A_{2})\varphi_{x_{1}},1\rangle{\rm P}_{\varphi}(dx_{1})\\
&=\langle\mathfrak{I}_{2}(A_{2})\circ\mathfrak{I}_{1}(A_{1})\varphi,1\rangle\\
&=\langle\mathfrak{I}(A_{1}\times A_{2})\varphi,1\rangle.
\end{align}
Hence by Carath${\rm\acute{e}}$odory theorem ${\rm P}_{\varphi}\times{\rm P}_{\varphi_{x_{1}}}$ is indeed $\langle\mathfrak{I}\varphi,1\rangle$.
Therefore
\begin{align}
\int_{A_{1}}{\rm P}_{\varphi}(dx_{1})\int_{A_{2}}\langle\varphi_{x_{1}x_{2}},b\rangle{\rm P}_{\varphi_{x_{1}}}(dx_{2})
&=\int_{A_{1}\times A_{2}}\langle\varphi_{x},b\rangle\langle\mathfrak{I}(dx)\varphi,1\rangle\\
&=\int_{A_{1}}{\rm P}_{\varphi}(dx_{1})\int_{A_{2}}\langle\varphi_{(x_{1},x_{2})},b\rangle{\rm P}_{\varphi_{x_{1}}}(dx_{2})
\end{align}
for all $A_{1}\in\mathscr{B}(X_{1})$ and $A_{2}\in\mathscr{B}(X_{2})$, which implies that there is a nonnegative $\mathscr{B}(X_{1}\times X_{2})$ measurable function that is equal to $\langle\varphi_{(x_{1},x_{2})},b\rangle$ $\langle\mathfrak{I}\varphi,1\rangle$-a.s. and agrees with $\langle\varphi_{x_{1}x_{2}},b\rangle$ on a set $M$ satisfying for ${\rm P}_{\varphi}$-a.s. $x_{1}$ the section $M^{c}_{x_{1}}$ of $M^{c}$ is a ${\rm P}_{\varphi_{{x}_{1}}}$-null set. Since ${\rm P}_{\varphi}$ and ${\rm P}_{\varphi_{x_{1}}}$ are finite, this is equivalent to $M^{c}$ is a $({\rm P}_{\varphi}\times{\rm P}_{\varphi_{x_{1}}})$-null set. Thus $\langle\varphi_{x_{1}x_{2}},b\rangle$ is $\langle\mathfrak{I}\varphi,1\rangle$ measurable for all $0\leq b\in\mathscr{M}$ (hence for all $b\in\mathscr{M}$).
(\romannumeral3) Apparently it holds that 
\begin{align}
        \int_{A}\langle\varphi_{x_{1}x_{2}},b\rangle\langle\mathfrak{I}(dx)\varphi,1\rangle=\langle\mathfrak{I}(A)\varphi,b\rangle
\end{align}
for all $b\in\mathscr{M}$ and $A\in\mathscr{B}(X_{1}\times X_{2})$.
\end{proof}
This gives a way to find out $\{\varphi_{(x_{1},x_{2})}:(x_{1},x_{2})\in X_{1}\times X_{2}\}$ by iterating $\{\varphi_{x_{1}}:x_{1}\in X_{1}\}$ and $\{\varphi_{x_{2}}:x_{2}\in X_{2}\}$ sequentially.

However, in general the composition of two instruments satisfies $[\mathfrak{I}_{1}\circ\mathfrak{I}_{2},\mathfrak{I}_{2}\circ\mathfrak{I}_{1}]\neq 0$ so that the sequential update of a prior normal state according to some instruments in different orders varies. 

In classical Bayesian inference, $A\mapsto K(\theta,A)1(\theta)$ is a joint observable of the observables $A_{i}\mapsto K_{i}(\theta,A_{i})1(\theta),i\in[n]$ in the sense that $K(\theta,A)1(\theta)=\left[\prod_{i=1}^{n}K_{i}(\theta,A_{i})\right]1(\theta)$ for all measurable rectangles $\prod_{i=1}^{n} A_{i}$. This means that we can measure the observables $A_{i}\mapsto K_{i}(\theta,A_{i})1(\theta),i\in[n]$ jointly by measuring them sequentially. However, this often fails in the quantum case since $\mathfrak{I}^{*}1$ is a joint observable of the observables $\mathfrak{I}_{1}^{*}(X_{1})\circ\cdots\circ\mathfrak{I}_{i-1}^{*}(X_{i-1})\circ\mathfrak{I}_{i}^{*}(\cdot)\circ\mathfrak{I}_{i+1}^{*}(X_{i+1})\circ\cdots\circ\mathfrak{I}_{n}^{*}(X_{n})1,i\in[n]$ and generally $\mathfrak{I}_{1}^{*}(X_{1})\circ\cdots\circ\mathfrak{I}_{i-1}^{*}(X_{i-1})\circ\mathfrak{I}_{i}^{*}(\cdot)\circ\mathfrak{I}_{i+1}^{*}(X_{i+1})\circ\cdots\circ\mathfrak{I}_{n}^{*}(X_{n})1\neq\mathfrak{I}_{i}^{*}1$ for $i\in\{2,\cdots,n\}$.
A sufficient condition that one can measure the observables $\mathfrak{I}_{i}^{*}{1},i\in[n]$ jointly by measuring them sequentially is given below and its proof is straight forward.
\begin{proposition}
   If $\mathfrak{I}_{i}^{*}$ commutes with $\mathfrak{I}_{1}^{*}\circ\cdots\circ\mathfrak{I}_{i-1}^{*}$ for any $i\in\{2,\cdots,n\}$ then 
   the observables $\nu_{i},i\in[n]$ induced by $\mathfrak{I}_{i},i\in[n]$ have $\mathfrak{I}^{*}1$ as a joint observable, where $\mathfrak{I}$ is the composition of $\mathfrak{I}_{i},i\in[n]$.
\end{proposition}

\section{Limit of Posterior Normal State}\label{sec4}
In classical Bayesian inference, under some regularity conditions the posterior distribution is asymptotically normal, i.e. it can be approximated by an appropriate normal distribution when the number of independently identically distributed observations is sufficiently large. However, if a quantum system $S$ described by a von Neumann algebra $\mathscr{M}$ is in the prior normal state $\varphi$ before sufficiently many sequentially performed measurements with apparatuses corresponding to instruments $\mathfrak{I}_{i},i\in[n]$, then after these measurements $S$ will probably not be in a converging posterior normal state. For example, identify $\mathcal{B}(\mathbb{H})_{*}$ with $\mathcal{B}_{1}(\mathbb{H})$ and
define an instrument $\mathfrak{I}\in{\rm Ins}(\{\varnothing,X\},\mathcal{B}_{1}(\mathbb{H}))$ by 
\begin{align}
    \mathfrak{I}(X)\cdot=u(\cdot)u^{*},
\end{align}
where $\{\varnothing,X\}$ is the trivial $\sigma$-algebra on $X$ and $u\in\mathcal{B}(\mathbb{H})$ a unitary operator. We assert that $\{\rho_{x}=u\rho u^{*}:x\in X\}$ is the unique family of posterior normal states w.r.t. $(\mathfrak{I},\rho)$. By repeatedly applying $\mathfrak{I}$ to the prior normal state $\rho$ one sees that the posterior normal state $\rho_{x}$ does not converge. However, in some cases the posterior normal state converges as the sample size goes to infinity.

Let $\mathfrak{I}:\{\varnothing,X\}\to\mathcal{B}^{+}({\rm M}_{k}(\mathbb{C}))$ be a CP instrument where ${\rm M}_{k}(\mathbb{C})$ is the von Neumann algebra of $k\times k$ complex matrices.
\begin{theorem}\label{the4}
    Assume that the spectrum of $\mathfrak{I}(X)$ has the unique eigenvalue $\{1\}$ on the unit circle and this eigenvalue is simple.
    Then there are positive numbers $C,\alpha>0$ such that for any $n\in\mathbb{N}_{+}$ we have
    \begin{align}
       \lVert\rho_{x_{1}\cdots x_{n}}-\rho_{*}\rVert_{1}\leq Ce^{-\alpha n},\ \forall \rho\in\mathcal{G}(\mathbb{C}^{k}),
    \end{align}
    where $\rho_{x_{1}\cdots x_{n}}=\mathfrak{I}^{n}\rho$ is the posterior normal state after performing the sequential measurement scheme $\mathfrak{I}^{n}$ and $\rho_{*}\in\mathcal{G}(\mathbb{C}^{k})$ satisfies $\mathfrak{I}(X)\rho_{*}=\rho_{*}$.
\end{theorem}
\begin{proof}
    We first show that $\{\varphi_{x}=\mathfrak{I}(X)\varphi:x\in X\}$ is the unique family of posterior normal states w.r.t. $(\mathfrak{I},\varphi)$, where $\mathfrak{I}:\{\varnothing,X\}\to\mathcal{B}^{+}(\mathscr{M}_{*})$ is an instrument and $\varphi$ is a prior normal state. Assume that $\{\varphi_{x}:x\in X\}$ is a family of posterior normal states w.r.t. $(\mathfrak{I},\varphi)$. Then for all $a\in\mathscr{M}$, the function $x\mapsto\langle\varphi_{x},a\rangle$ is ${\rm P}_{\varphi}$ measurable and the following equation holds
    \begin{align}
            \int_{X}\langle\varphi_{x},a\rangle{\rm P}_{\varphi}(dx)=\langle\mathfrak{I}(X)\varphi,a\rangle\text{.}
    \end{align}
    Since the probability space $(X,\{\varnothing,X\},{\rm P}_{\varphi})$ is complete, the function $x\mapsto\langle\varphi_{x},a\rangle$ is $\{\varnothing,X\}$ measurable for all $a\in\mathscr{M}$. Therefore $x\mapsto\langle\varphi_{x},a\rangle$ shall be constant for all $a\in\mathscr{M}$. This implies that for all $x\in X$, $\varphi_{x}=\varphi_{0}$ for some $\varphi_{0}\in\mathfrak{S}(\mathscr{M})$. Again due to the arbitrariness of $a$, $\varphi_{0}=\mathfrak{I}(X)\varphi$.
    The rest of the proof follows from ${\rm \cite{bib23}}$.
\end{proof}
Let $(\Omega,\mathscr{F},{\rm P})$ be a probability space and $\mathfrak{T}:\Omega\to\Omega$ an invertible, measure preserving, ergodic map (i.e for all $F\in\mathscr{F}$ satisfying $\mathfrak{T}^{-1}(F)=F$ we have ${\rm P}(F) = 0$ or $1$).
We say that a map $f\in\mathcal{B}^{+}({\rm M}_{k}(\mathbb{C}))$ is strictly positive iff $f(a)>0$ for all $a\in{\rm M}_{k}^{+}(\mathbb{C})\setminus\{0\}$. Let $f_{0}$ be a ${\rm CPIns}(\{\varnothing,X\},{\rm M}_{k}(\mathbb{C}))$-valued $\mathscr{F}$ measurable map. Define $f_{n}(\omega)=f_{0}[\mathfrak{T}^{n}(\omega)],\forall\omega\in \Omega,n\in\mathbb{Z}$.
\begin{theorem}\label{the5}
    Assume that there is a positive integer $N$ such that 
    \begin{align}
        &(1)\ {\rm P}[f_{N}(X)\circ\cdots\circ f_{0}(X)\ \text{is strictly positive}]> 0;\\
        &(2)\ {\rm P}[\ker(f_{0}^{*}(X))\cap{\rm M}_{k}^{+}(\mathbb{C})=\{0\}]=1.
    \end{align}
    Then there is a positive number $\alpha>0$, 
    a $\mathcal{G}(\mathbb{C}^{k})$ valued $\mathscr{F}$ measurable map $\rho_{0}$ and a sequence of $\mathscr{F}$ measurable maps $\{\rho_{n}=\rho_{0}[\mathfrak{T}^{n}(\omega)],\forall \omega\in\Omega,n\in\mathbb{Z}\}$ such that 
    for any $(m,z,n)\in\mathbb{Z}^{3}$ satisfying $m<n,m\leq z,z\leq n$ we have
    \begin{align}
        \lVert\rho_{x_{m}\cdots x_{n}}-\rho_{n}\rVert_{1}\leq C_{\alpha,z}e^{-\alpha(n-m)},\ \forall \rho\in\mathcal{G}(\mathbb{C}^{k}),
    \end{align}
    where $\rho_{x_{m}\cdots x_{n}}=f_{n}(X)\circ\cdots\circ f_{m}(X)\rho$ is the posterior normal state after performing the sequential measurement scheme $f_{n}\circ\cdots\circ f_{m}$ and $C_{\alpha,z}$ is finite ${\rm P}$-a.s.
\end{theorem}
\begin{proof}
Since a CP instrument $\mathfrak{I}:\{\varnothing,X\}\to\mathcal{B}^{+}(\mathscr{M}_{*})$ is uniquely determined by its channel $\mathfrak{I}(X)$, there is a natural one-to-one correspondence between CP instruments $\mathfrak{I}:\{\varnothing,X\}\to\mathcal{B}^{+}({\rm M}_{k}(\mathbb{C}))$ and channels $\Psi\in\mathcal{C}({\rm M}_{k}(\mathbb{C}))$. The rest of the proof follows from Lemma 3.8, Lemma 3.11 and Lemma 3.14 in \cite{bib22}.
\end{proof}
In classical Bayesian inference, a Bayesian is able to learn about an unknown parameter $\theta$ by updating her prior, which to some extent, can be viewed as her prior knowledge about $\theta$, to her posterior. In other words, if an oracle were to know the true value of the parameter, a Bayesian shall ensure that with enough observations she would get close to this true value. This is guaranteed by the weak consistency of posterior distribution. Let $X_{n}$ be a random $n$-tuple (i.i.d. observations) whose conditional distribution is ${\rm P}_{\theta}^{n}$ (an $n$-fold product of ${\rm P}_{\theta}$), $\Pi$ a prior distribution and $\Pi_{X_{n}}$ the posterior distribution. Assume that the true value of $\theta$ is $\theta_{0}$ (iff the ground truth distribution of $\theta$ is the Dirac measure $\Pi_{0}=1_{(\cdot)}(\theta_{0})$). Then under some conditions, $\Pi_{X^{(n)}}$ converges weakly to $\Pi_{0}$ in probability ${\rm P}^{n}_{\theta_{0}}$, as $n\to\infty$.

In the quantum case, the weak consistency of posterior normal state might be defined as follows. Let $\varphi$ be a prior normal state that a quantum Bayesian stands for and $\varphi_{0}$ the ground truth. Let $\left\{\mathfrak{I}_{n}\in\text{Ins}\left(\prod_{i=1}^{n}\mathscr{A}_{i},\mathscr{M}_{*}\right)\right\}_{n=1}^{\infty}$ be a sequence of instruments such that $\langle\mathfrak{I}_{n+1}(A^{(n)}\times X_{n+1})\varphi,1\rangle=\langle\mathfrak{I}_{n}(A^{(n)})\varphi,1\rangle$ for all $A^{(n)}\in\prod_{i=1}^{n}\mathscr{A}_{i}$ and $n\in\mathbb{N}_{+}$. Denote by $\{\varphi_{x^{(n)}}:x^{(n)}\in\Pi_{i=1}^{n}X_{n}\}$ (resp. $\{\varphi_{0,x^{(n)}}:x^{(n)}\in\Pi_{i=1}^{n}X_{n}\}$) a family of posterior normal states w.r.t. $(\mathfrak{I}_{n},\varphi)$ (resp. $(\mathfrak{I}_{n},\varphi_{0})$). Then under some conditions $\varphi_{X^{(n)}}$ could converge to $\varphi_{0,X^{(n)}}$ with respect to an appropriate metrizable topology in probability $\langle\mathfrak{I}_{n}\varphi_{0},1\rangle$.

Assume that $X$ consists of a single element and equipped with the trivial topology. By the definition of the weak consistency of posterior normal state, Theorem \ref{the2} and Theorem \ref{the4}, we have
\begin{align}
\lim_{n\to\infty}\text{tr}(\mathfrak{I}^{n}\rho_{0})\left(\lVert\rho_{X_{1}\cdots X_{n}}-\rho_{0,X_{1}\cdots X_{n}}\rVert_{1}>\epsilon\right)=0,\ \forall\epsilon>0,\forall\rho,\rho_{0}\in\mathcal{G}(\mathbb{C}^{k}).
\end{align}
In other words, given sequential measurement scheme $\left\{\mathfrak{I}^{n}\right\}_{n=1}^{\infty}$, for any prior normal state $\rho\in\mathcal{G}(\mathbb{C}^{k})$, the posterior normal state $\rho_{X_{1}\cdots X_{n}}$ is weakly consistent at any $\rho_{0}\in\mathcal{G}(\mathbb{C}^{k})$. 

Similarly, by the definition of the weak consistency of posterior normal state, Theorem \ref{the2} and Theorem \ref{the5}, we have
\begin{align}
\lim_{(n-m)\to\infty}\text{tr}(f^{n-m}\rho_{0})\left(\lVert\rho_{X_{m}\cdots X_{n}}-\rho_{0,X_{m}\cdots X_{n}}\rVert_{1}>\epsilon\right)=0,\ {\rm P}\text{-a.s.},\forall\epsilon>0,\forall\rho,\rho_{0}\in\mathcal{G}(\mathbb{C}^{k}),
\end{align}
where $f^{n-m}=f_{n}\circ\cdots\circ f_{m}$. 
In other words, given sequential measurement scheme $f^{n-m},(m,n)\in\mathbb{Z}^{2},m<n$, for any prior normal state $\rho\in\mathcal{G}(\mathbb{C}^{k})$, the posterior normal state $\rho_{X_{m}\cdots X_{n}}$ is weakly consistent at any $\rho_{0}\in\mathcal{G}(\mathbb{C}^{k})$ almost surely.

Let $\mathbb{H}$ be a separable complex Hilbert space.
Define a CP instrument 
$\mathfrak{I}\in{\rm CPI}{\rm ns}(\mathscr{A},\mathcal{B}_{1}(\mathbb{H}))$ by
\begin{align}
    \mathfrak{I}(A)\cdot=\sum_{i\in A}a_{i}(\cdot)a_{i}^{*},\ \forall A\in\mathscr{A},
\end{align}
where $a_{i}\in\mathcal{B}(\mathbb{H})$ for all $i\in\mathbb{N}_{+}$ and $\sum_{i=1}^{\infty}a_{i}^{*}a_{i}=1$ in the strong operator topology of 
$\mathcal{B}(\mathbb{H})$. Define a Markov kernel $K:\mathcal{G}(\mathbb{H})\times\mathscr{B}(\mathcal{G}(\mathbb{H}))\to[0,1]$ by
\begin{align}
  (\rho,G)\mapsto\sum_{i=1}^{\infty}{\rm tr}(a_{i}\rho a_{i}^{*})1_{G}(\rho_{i}),
\end{align}
where $\rho_{i}=a_{i}\rho a_{i}^{*}/{\rm tr}(a_{i}\rho a_{i}^{*})$. Let $(\Omega,\mathscr{F},{\rm P})$ be a probability space and $\{X_{n}:\Omega\to\mathcal{G}(\mathbb{H})\}_{n=0}^{\infty}$ a Markov chain such that ${\rm P}(X_{n+1}\in G|X_{n}=\rho)=K(\rho,G),\ \forall n\in\mathbb{N}$.
Denote by ${\rm P}_{\rho_{0}}$ the probability measure of $\{X_{n}\}_{n=0}^{\infty}$ initializing at $\rho_{0}$, i.e. ${\rm P}_{\rho_{0}}(F)={\rm P}(F|X_{0}=\rho_{0})$ for all $F\in\mathscr{F}$.
\begin{theorem}{\rm\cite{bib21}}
     Let $\{X_{n}\}_{n=0}^{\infty}$ be the Markov chain starting at $\rho_{0}$. Then there is a $\{\rho\in\mathcal{G}(\mathbb{H}):\sum_{i=1}^{\infty}a_{i}\rho a_{i}^{*}=\rho\}$ valued random variable $X_{\rho_{0}}$ such that
    \begin{align}
        \cfrac{1}{n}\sum_{i=1}^{n}X_{n}\to X_{\rho_{0}},\ \ \text{weakly}^{*}\text{ and ${\rm P}_{\rho_{0}}$-a.s.}
    \end{align}
    as $n\to\infty$.
\end{theorem} 

\section{Quantum Bayesian Inference}\label{sec5}
With quantum Bayes' rule, a quantum analogue of Bayesian inference is ready to debut. A quantum statistical model is a pair $(\nu,\Lambda)$
    where $\nu:\mathscr{A}\to\mathscr{M}$ is an observable
    and $\Lambda\subseteq\mathfrak{S}(\mathscr{M})$ a set of normal states.
    Moreover, if $\Lambda=\{\varphi(\theta)\in\mathfrak{S}(\mathscr{M}):\theta\in\Theta\}$, then $(\nu,\varphi(\theta))$ is called a parametric quantum statistical model. 
    
    A quantum Bayesian decision problem consists of five elements: a parameter space $(\Theta,\mathscr{E},\Pi)$, which is a probability space; a random experiment $(\nu,\varphi(\theta))$, which is a parametric quantum statistical model; an action space $(Y,\mathscr{B})$, which is a measurable space; a randomized decision rule $\delta:X\times\mathscr{B}\to[0,1]$, which is a Markov kernel; a loss function $L:\Theta\times Y\to\mathbb{R}_{+}$ such that $L(\theta,y)$ is $\mathscr{B}$ measurable for all $\theta\in\Theta$. Denote by $\Delta$ a class of randomized decision rule $\delta:X\times\mathscr{B}\to[0,1]$. If a randomized decision rule $\delta$ satisfying $\delta(x,\cdot)$ is a Dirac measure for all $x\in X$, then $\delta$ is called a non-randomized decision rule or decision rule in brief. Denote by $\mathcal{D}$ a class of decision rule $\delta:X\times\mathscr{B}\to[0,1]$. The risk function $R:\Theta\times\Delta\to\overline{\mathbb{R}}_{+}$ of a quantum Bayesian decision problem is defined by
        \begin{align}
        R(\theta,\delta)={\rm E}_{x}{\rm E}_{y}L(\theta,y)
        &=\int_{X}\langle\varphi(\theta),\nu(dx)\rangle\int_{Y}L(\theta,y)\delta(x,dy)\text{.}
        \end{align}
    The quantum Bayes risk ${\rm R}_{\Pi}:\Delta\to\overline{\mathbb{R}}_{+}$ of a quantum Bayesian decision problem is defined by
        \begin{align}
            {\rm R}_{\Pi}(\delta)={\rm E}_{\theta}R(\theta,\delta)=\int_{\Theta}\Pi(d\theta)
            \int_{X}\langle\varphi(\theta),\nu(dx)\rangle\int_{Y}L(\theta,y)\delta(x,dy).
        \end{align}
    Moreover, if there is a randomized decision rule $\check{\delta}\in\Delta$ such that 
            \begin{align}
                {\rm R}_{\Pi}(\check{\delta})=\inf_{\delta\in\Delta}{\rm R}_{\Pi}(\delta)\text{,}
            \end{align}
            then $\check{\delta}$ is called a quantum Bayes solution w.r.t. $(\Pi,\Delta)$.
    
    In the following we would like to explore the admissibility of quantum Bayesian solution. We say that a randomized decision rule $\delta$ is inadmissible iff there is another randomized decision rule $\delta'\in\Delta$ such that (\romannumeral1)
        $R(\theta,\delta')\leq R(\theta,\delta)$ for all $\theta\in\Theta$;
        (\romannumeral2) There is a $\theta'\in\Theta$ such that $R(\theta',\delta')<R(\theta',\delta)$.
    \begin{theorem}
        Assume that $\Theta$ is a metric space and $\mathscr{E}$ the Borel $\sigma$-algebra on $\Theta$.
        Denote by $\check{\delta}$ a quantum Bayes solution w.r.t. $(\Pi,\Delta)$.
        If 

            (\romannumeral1) $\Pi(E)>0$ for all open subsets $E$ of $\Theta$;

            (\romannumeral2) ${\rm R}_{\Pi}(\check{\delta})<\infty$;

            (\romannumeral3) $R(\theta,\delta)$ is a continuous function of $\theta$ for all $\delta\in\Delta$,\\
            then $\check{\delta}$ is admissible.
    \end{theorem}
    \begin{proof}
        If $\check{\delta}$ is inadmissible, 
        then there is a randomized decision rule $\delta'\in\Delta$ such that
        \begin{align}
            &R(\theta,\delta')\leq R(\theta,\check{\delta}),\ \forall\theta\in\Theta;\\
            &\exists\ \theta'\in\Theta\ \text{s.t.}\ R(\theta',\delta')<R(\theta',\check{\delta})\text{.}
        \end{align}
        By condition (\romannumeral3), there is a positive number $\epsilon>0$ together with an open ball $S_{\epsilon}(\theta')$
        such that 
        \begin{align}
            R(\theta,\delta')< R(\theta,\check{\delta})-\epsilon,\ \forall\theta\in S_{\epsilon}(\theta').
        \end{align}
        Then we have 
        \begin{align}
            {\rm R}_{\Pi}(\delta')&=\int_{S_{\epsilon}(\theta')}R(\theta,\delta')\Pi(d\theta)
            +\int_{\Theta\setminus S_{\epsilon}(\theta')}R(\theta,\delta')\Pi(d\theta)\\
            &<\int_{S_{\epsilon}(\theta')}[R(\theta,\check{\delta})-\epsilon]\Pi(d\theta)
            +\int_{\Theta\setminus S_{\epsilon}(\theta')}R(\theta,\check{\delta})\Pi(d\theta)\\
            &={\rm R}_{\Pi}(\check{\delta})-\epsilon\Pi[S_{\epsilon}(\theta')]\\
            &<{\rm R}_{\Pi}(\check{\delta}),
        \end{align}
        which contradicts that $\check{\delta}$ is a quantum Bayes solution w.r.t. $(\Pi,\Delta)$.        
    \end{proof} 
    \begin{theorem}
        If $\check{\delta}$ is the unique quantum Bayes solution w.r.t. $(\Pi,\Delta)$,
        then $\check{\delta}$ is admissible.
    \end{theorem}
    \begin{proof}
        If $\check{\delta}$ is inadmissible, 
        then there is a randomized decision rule $\delta'\in\Delta$ such that
        \begin{align}
            &R(\theta,\delta')\leq R(\theta,\check{\delta}),\ \forall\theta\in\Theta;\\
            &\exists\ \theta'\in\Theta\ \text{s.t.}\ R(\theta',\delta')<R(\theta',\check{\delta})\text{.}
        \end{align}
        Then we have 
        \begin{align}
            {\rm R}_{\Pi}(\delta')&=\int_{\Theta}R(\theta,\delta')\Pi(d\theta)
            \leq\int_{\Theta}R(\theta,\check{\delta})\Pi(d\theta)
            ={\rm R}_{\Pi}(\check{\delta}),
        \end{align}
        which shows that $\delta'$ is a quantum Bayes solution w.r.t. $(\Pi,\Delta)$ as well. But
        this contradicts that $\check{\delta}$ is the unique quantum Bayes solution w.r.t. $(\pi,\Delta)$.     
    \end{proof}
    We say that $\tilde{\delta}\in\mathcal{D}$ is a minimax decision rule iff it satisfies
        \begin{align}
            \sup_{\theta\in\Theta}R(\theta,\tilde{\delta})=\inf_{\delta\in\mathcal{D}}\sup_{\theta\in\Theta}R(\theta,\delta)\text{.}
        \end{align}
    \begin{theorem}
        Let $\check{\delta}$ be a quantum Bayes solution w.r.t. $(\Pi,\mathcal{D})$.
        If $R(\theta,\check{\delta})=c$ for all $\theta\in\Theta$, then $\check{\delta}$ is a minimax decision rule.
    \end{theorem}
    \begin{proof}
        On the one hand,
        \begin{align}
        c={\rm R}_{\Pi}(\check{\delta})&=\inf_{\delta\in\mathcal{D}}{\rm R}_{\Pi}(\delta)\\
        &\leq\sup_{\Pi'\in\mathcal{P}}\inf_{\delta\in\mathcal{D}}{\rm R}_{\Pi'}(\delta)\\
        &\leq\inf_{\delta\in\mathcal{D}}\sup_{\Pi'\in\mathcal{P}}{\rm R}_{\Pi'}(\delta)\\
        &\leq\inf_{\delta\in\mathcal{D}}\sup_{\theta\in\Theta}R(\theta,\delta),
    \end{align}
    where $\mathcal{P}$ is a class of prior distributions such that $\Pi\in\mathcal{P}$.
On the other hand, 
\begin{align}
    c=\sup_{\theta\in\Theta}R(\theta,\check{\delta})\ge\inf_{\delta\in \mathcal{D}}\sup_{\theta\in\Theta}R(\theta,\delta).
\end{align}
    \end{proof}
    \begin{theorem}\label{The5.4}
        Let $\{\Pi_{n}\}_{n=1}^{\infty}$ be a sequence of prior distributions and $\check{\delta}_{n}$ a quantum Bayes solution w.r.t. $(\Pi_{n},\mathcal{D})$.
        Let $\delta\in\mathcal{D}$ be a decision rule.
        If 
        \begin{align}\label{Inequ66}
            \sup_{\theta\in\Theta}R(\theta,\delta)\leq\limsup_{n\to\infty}{\rm R}_{\Pi_{n}}(\check{\delta}_{n}),
        \end{align}
        then $\delta$ is a minimax decision rule. 
    \end{theorem}
    \begin{proof}
        If $\delta$ is not a minimax decision rule, then there is a decision rule $\delta'\in\mathcal{D}$ such that 
        \begin{align}
            \sup_{\theta\in\Theta}R(\theta,\delta')<\sup_{\theta\in\Theta}R(\theta,\delta).
        \end{align}
        For any $n\ge1$ we have
        \begin{align}
            {\rm R}_{\Pi_{n}}(\check{\delta}_{n})&\leq{\rm R}_{\Pi_{n}}(\delta')\\
            &\leq\sup_{\theta\in\Theta}R(\theta,\delta')\\
            &<\sup_{\theta\in\Theta}R(\theta,\delta),
        \end{align}
        which contradicts the condition (\ref{Inequ66}).        
    \end{proof}
    
    Let $\varphi\in\mathfrak{S}(\mathscr{M})$ be a prior normal state, $\lambda:\mathscr{E}\to\mathscr{M}$ an observable, $\mathfrak{I}:\mathscr{A}\to\mathcal{B}^{+}(\mathscr{M}_{*})$ an instrument and $\{\varphi_{x}:x\in X\}$ a family of posterior normal states w.r.t. $(\mathfrak{I},\varphi)$.
\begin{definition}
    The quantum posterior risk ${\rm R}_{\varphi}:\Delta\to\overline{\mathbb{R}}_{+}$ of a quantum Bayesian decision problem is defined by
    \begin{align}
       {\rm R}_{\varphi}(\delta)={\rm E}_{\theta\lvert x}{\rm E}_{y}L(\theta,y)=\int_{\Theta}\langle\varphi_{x},\lambda(d\theta)\rangle\int_{Y}L(\theta,y)\delta(x,dy).
    \end{align}
        Moreover, if there is a randomized decision rule $\delta^{*}\in\Delta$ such that 
        \begin{align}
            {\rm R}_{\varphi}(\delta^{*})=\inf_{\delta\in\Delta}{\rm R}_{\varphi}(\delta)\text{,}
        \end{align}
        then $\delta^{*}$ is called a quantum posterior solution w.r.t. $(\varphi,\Delta)$.
\end{definition}
Apparently our framework retains the classical one as a special case.
In classical statistics, a statistical model is a triad $(X,\mathscr{A},\mathcal{P})$ where $(X,\mathscr{A})$ is a measurable space and $\mathcal{P}$ a class of probability measures on $(X,\mathscr{A})$. Moreover, if $\mathcal{P}=\{{\rm P}_{\theta}:\theta\in\Theta\}$, then $(X,\mathscr{A},{\rm P}_{\theta})$ is called a parametric statistical model. If we replace the parametric quantum statistical model $(\nu,\varphi(\theta))$ in a quantum Bayesian decision problem with a parametric statistical model $(X,\mathscr{A},{\rm P}_{\theta})$ where ${\rm P}_{\theta}$ is a Markov kernel, we will get a classical Bayesian decision problem. 

It is noted that in classical Bayesian decision, a posterior solution w.r.t. $(\Pi,\Delta)$ is always a Bayes solution w.r.t. $(\Pi,\Delta)$ and each Bayes solution w.r.t. $(\Pi,\Delta)$ minimizes the posterior risk for ${\rm P}=\int{\rm P}_{\theta}d\Pi$ almost all $x\in X$. However, this is not true in the quantum case. To see this, assume that the map $(\theta,x)\mapsto\int_{Y} L(\theta,y)\delta(x,dy)$ is $\mathscr{E}\times\mathscr{A}$ measurable for all $\delta\in\Delta$ and $\langle\varphi(\theta),\nu(A)\rangle$ is a Markov kernel defined by $\int_{A}p(x|\theta)\mu(dx),\ \forall A\in\mathscr{A}$, where $p(x|\theta)$ is a nonnegative $\mathscr{E}\times\mathscr{A}$ measurable real valued function satisfying $\int_{X}p(x\lvert\theta)\mu(dx)=1,\ \forall\theta\in\Theta$, and $\mu$ is a $\sigma$-finite measure on $(X,\mathscr{A})$. Then each quantum Bayes solution w.r.t. $(\Pi=\int\pi d\nu,\Delta)$ shall minimize the posterior risk $\int_{\Theta}\pi(\theta|x)\nu(d\theta)\int_{Y}L(\theta,y)\delta(x,dy)$ for ${\rm P}$ almost all $x\in X$, where $\pi\in L^{1}_{+,1}(\nu)$ and $\nu$ is a $\sigma$-finite measure on $(\Theta,\mathscr{E})$. Since probability measures $\langle\varphi_{x},\lambda(\cdot)\rangle$ and $\int_{(\cdot)}\pi(\theta|x)\nu(d\theta)$ can be quite different, it is likely to happen that there is a quantum Bayes solution w.r.t. $(\Pi,\Delta)$ not minimizing the quantum posterior risk $\int_{\Theta}\langle\varphi_{x},\lambda(d\theta)\rangle\int_{Y}L(\theta,y)\delta(x,dy)$ for ${\rm P}$ almost all $x\in X$.

In the following we will concentrate on a quantum analogue of Bayesian inference. This topic mainly consists of four parts: quantum posterior point estimation, quantum posterior credible interval, quantum posterior hypothesis testing and quantum posterior prediction. One could see that in terms of framework, quantum Bayesian inference is almost the same as the classical one.

Assume that $(X,\mathscr{A}),(Y,\mathscr{B})$ and $(\Theta,\mathscr{E})$ are all $(\mathbb{R},\mathscr{B}(\mathbb{R}))$. 
Denote by ${\rm P}_{\varphi_{x}}$ the probability measure 
$\langle\varphi_{x},\lambda(\cdot)\rangle$ and by ${\rm F}(\theta\lvert x)$ the cumulative distribution function corresponding to ${\rm P}_{\varphi_{x}}$.
\begin{theorem}
    Suppose the loss function $L(\theta,y)=c(\theta)(\theta-y)^{2}$, where $c(\theta)\ge 0$ for all $\theta\in\Theta$. 
    Then the quantum posterior solution $\delta^{*}$ w.r.t. $(\varphi,\mathcal{D})$ is
    \begin{align}
        \delta^{*}(x,B)=1_{B}\left[\frac{{\rm E}_{\theta|x}\theta c(\theta)}{{\rm E}_{\theta|x}c(\theta)}\right]\text{.}
    \end{align}
    Specifically, if $c(\theta)=1$ for all $\theta\in\Theta$, then the quantum posterior solution $\delta^{*}$ w.r.t. $(\varphi,\mathcal{D})$ 
    is
    \begin{align}
        \delta^{*}(x,B)=1_{B}({\rm E}_{\theta\lvert x}\theta)\text{.}
    \end{align}
\end{theorem}
\begin{proof}
    Straight calculation shows that the quantum posterior risk 
    \begin{align}
        {\rm R}_{\varphi}(\delta)&={\rm E}_{\theta\lvert x}{\rm E}_{y}c(\theta)(\theta-y)^{2}\\
        &={\rm E}_{\theta\lvert x}c(\theta)[\theta-\delta(x)]^{2},
    \end{align}
    where $\delta(x)\in Y$ is the point such that $\delta(x,B)=1_{B}[\delta(x)],\ \forall B\in\mathscr{B}$.
    Apparently ${\rm R}_{\varphi}(\delta)$ takes its minimum when
    \begin{align}   
        \delta(x)={\rm E}_{\theta\lvert x}\theta c(\theta)/{\rm E}_{\theta\lvert x}c(\theta).
    \end{align}
\end{proof}
\begin{theorem}
    Suppose the loss function $L(\theta,y)=\left\{
        \begin{aligned}
            k_{0}(\theta-y) & \ \ \ if\ \theta> y\\
            k_{1}(y-\theta) & \ \ \ if\ \theta\leq y
        \end{aligned}
        \right.
    $, where $k_{0},k_{1}\ge 0$.
    Then the quantum posterior solution $\delta^{*}$ w.r.t. $(\varphi,\mathcal{D})$
    is
    \begin{align}
        \delta^{*}(x,B)=1_{B}\left[\xi_{\theta\lvert x}\left(\cfrac{k_{0}}{k_{0}+k_{1}}\right)\right]\text{,}
    \end{align}
    where $\xi_{\theta\lvert x}(p):=\inf\{\theta\in\Theta:{\rm F}(\theta\lvert x)\ge p\}$ is the $p$ quantile of ${\rm F}(\theta|x)$.
    Specifically, if $k_{0}=k_{1}=1$, then the quantum posterior solution $\delta^{*}$ w.r.t. $(\varphi,\mathcal{D})$
    is
    \begin{align}
        \delta^{*}(x,B)=1_{B}[\xi_{\theta\lvert x}(0.5)]\text{.}
    \end{align}
\end{theorem}
\begin{proof}
    Let $p=k_{0}/(k_{0}+k_{1})$. 
    Assume that $y>\xi_{\theta\lvert x}(p)$.
    A little calculation shows that
    \begin{align}
        L[\theta,\xi_{\theta\lvert x}(p)]-L(\theta,y)=\left\{
            \begin{aligned}
                k_{1}[\xi_{\theta\lvert x}(p)-y] & \ \ \ if\ \theta\leq\xi_{\theta\lvert x}(p)\\
                (k_{0}+k_{1})\theta-k_{0}\xi_{\theta\lvert x}(p)-k_{1}y & \ \ \ if\ \xi_{\theta\lvert x}(p)<\theta\leq y\\
                k_{0}[y-\xi_{\theta\lvert x}(p)] & \ \ \ if\ \theta> y
            \end{aligned}
            \right.,
    \end{align}
    and 
    \begin{align}
        L[\theta,\xi_{\theta\lvert x}(p)]-L(\theta,y)&\leq k_{1}[\xi_{\theta\lvert x}(p)-y]1_{(-\infty,\xi_{\theta\lvert x}(p)]}\\
        &\ \ \ +k_{0}[y-\xi_{\theta\lvert x}(p)]1_{(\xi_{\theta\lvert x}(p),+\infty)}
    \end{align}
    since $(k_{0}+k_{1})\theta-k_{0}\xi_{\theta\lvert x}(p)-k_{1}y
    <(k_{0}+k_{1})y-k_{0}\xi_{\theta\lvert x}(p)-k_{1}y=k_{0}[y-\xi_{\theta\lvert x}(p)]$. Let $\delta_{0}\in\mathcal{D}$ such that $\delta_{0}(x)=\xi_{\theta\lvert x}(p)$.
    Then for any $\delta\in\mathcal{D}$ such that $\delta(x)>\xi_{\theta\lvert x}(p)$ we have
    \begin{align}
        {\rm R}_{\varphi}(\delta_{0})-{\rm R}_{\varphi}(\delta)
        &={\rm E}_{\theta\lvert x}\{L[\theta,\xi_{\theta\lvert x}(p)]-L[\theta,\delta(x)]\}\\
        &\leq k_{1}[\xi_{\theta\lvert x}(p)-\delta(x)]{\rm P}_{\varphi_{x}}[\theta\leq\xi_{\theta\lvert x}(p)]\\
        &\ \ \ +k_{0}[\delta(x)-\xi_{\theta\lvert x}(p)]{\rm P}_{\varphi_{x}}[\theta>\xi_{\theta\lvert x}(p)]\\
        &\leq\cfrac{k_{0}k_{1}}{k_{0}+k_{1}}[\xi_{\theta\lvert x}(p)-\delta(x)]\\
        &\ \ \ +\cfrac{k_{0}k_{1}}{k_{0}+k_{1}}[\delta(x)-\xi_{\theta\lvert x}(p)]=0,
    \end{align}
    which is equivalent to 
    \begin{align}
        {\rm R}_{\varphi}(\delta_{0})\leq{\rm R}_{\varphi}(\delta).
    \end{align}
    The inequality above also holds when $\delta(x)<\xi_{\theta\lvert x}(p)$
    and thus $\delta_{0}$ is the quantum posterior solution w.r.t. $(\varphi,\mathcal{D})$. 
\end{proof}
\begin{theorem}
    Suppose the loss function $L(\theta,y)=\left\{
        \begin{aligned}
            1 & \ \ \ if\ \lvert\theta-y\rvert>\epsilon\\
            0 & \ \ \ if\ \lvert\theta-y\rvert\leq\epsilon
        \end{aligned}
        \right.
    $ and ${\rm P}_{\varphi_{x}}(\cdot)=\int_{(\cdot)}p_{x}(\theta)\mu_{x}(d\theta)$ for a Borel measurable real valued function $p_{x}(\theta)$ and a $\sigma$-finite measure $\mu_{x}$ on $(\mathbb{R},\mathscr{B}(\mathbb{R}))$ and $p_{x}(\theta)$ has maximum value. Then as $\epsilon\to 0$ a quantum posterior solution $\delta^{*}$ w.r.t. $(\varphi,\mathcal{D})$
    is
    \begin{align}
        \delta^{*}(x,B)=1_{B}[\max_{\theta\in\Theta}p_{x}(\theta)]\text{.}
    \end{align}
\end{theorem}
\begin{proof}
    Straight calculation shows that 
    \begin{align}
        {\rm R}_{\varphi}(\delta)&={\rm E}_{\theta\lvert x}{\rm E}_{y}L(\theta,y)\\
        &=1-\int_{\delta(x)-\epsilon}^{\delta(x)+\epsilon}p_{x}(\theta)\mu_{x}(d\theta).
    \end{align}
    Apparently minimizing ${\rm R}_{\varphi}(\delta)$ is equivalent to maximizing
    $\int_{\delta(x)-\epsilon}^{\delta(x)+\epsilon}p_{x}(\theta)\mu_{x}(d\theta)$ 
    and thus $\delta(x)=\max_{\theta\in\Theta}p_{x}(\theta)$ as $\epsilon\to 0$.  
\end{proof}
We say that $\hat{\theta}_{{\rm E}}:={\rm E}_{\theta\lvert x}(\theta)$ is the quantum posterior mean estimator of $\theta$; $\hat{\theta}_{{\rm Q}}:=\xi_{\theta\lvert x}(0.5)$ is the quantum posterior median estimator of $\theta$ and $\hat{\theta}_{{\rm M}}:=\max_{\theta\in\Theta}p_{x}(\theta)$ is the quantum posterior mode estimator of $\theta$.
\begin{remark}
Under the premise of Theorem \ref{the4}, further assume that $\Theta$ is a metric space, $\mathscr{E}$ is the Borel $\sigma$-algebra on $\Theta$, $X$ consists of a single element and equipped with the trivial topology, and ${\rm tr}(\rho_{*}\lambda)$ is a Dirac measure. Then by Theorem \ref{the4} we have
\begin{align}
0\leq\lvert{\rm tr}[(\rho_{X_{1}\cdots X_{n}}-\rho_{*})\lambda]\rvert\leq\lVert\lambda\rVert\lVert\rho_{X_{1}\cdots X_{n}}-\rho_{*}\rVert_{1}\to 0,
\end{align}
as $n\to\infty$. This implies that ${\rm tr}(\rho_{X_{1}\cdots X_{n}}\lambda)$ converges weakly to ${\rm tr}(\rho_{*}\lambda)$ in probability ${\rm tr}(\mathfrak{I}^{n}\rho_{0})$. Thus $\hat{\theta}_{{\rm E}}$, $\hat{\theta}_{{\rm Q}}$ and $\hat{\theta}_{{\rm M}}$ are all weakly consistent.
\end{remark}
Assume that $Y$ is the set of closed intervals on the real line and the loss function 
$L(\theta,y)=k_{0}l(y)+k_{1}[1-1_{y}(\theta)]$, where $l(y)$ is the length of $y$ and $k_{0},k_{1}\ge 0$.
Then the quantum posterior risk
\begin{align}
    {\rm R}_{\varphi}(\delta)&={\rm E}_{\theta\lvert x}{\rm E}_{y}\{k_{0}l(y)+k_{1}[1-1_{y}(\theta)]\}\\
    &=k_{0}l[\delta(x)]+k_{1}{\rm E}_{\theta\lvert x}[1-1_{\delta(x)}(\theta)]\\
    &=k_{0}l[\delta(x)]+k_{1}{\rm P}_{\varphi_{x}}[\theta\notin\delta(x)].
\end{align}
Here a quantum posterior solution $\delta^{*}$ w.r.t. $(\varphi,\mathcal{D})$ is generally not easy to obtain. 
    A common strategy is to keep $l[\delta(x)]$ as small as possible while controlling ${\rm P}_{\varphi_{x}}[\theta\notin\delta(x)]$
    not to exceed a given small positive number.
	We say that a decision rule $\delta$ is a $1-\alpha$ quantum posterior credible interval iff
    \begin{align}
        {\rm P}_{\varphi_{x}}[\theta\in\delta(x)]\ge 1-\alpha,\ \forall x\in X,
    \end{align} 
    where $0<\alpha<1$.

Assume that $(Y,\mathscr{F})=([n],2^{[n]})$ and the loss function 
$L(\theta,y)=\sum_{i=1}^{n}\delta_{iy}1_{\Theta_{i}^{c}}(\theta)$, where $\delta_{(\cdot\cdot)}$ is the Kronecker delta and $\Theta_{i},i\in[n]$ are mutually disjoint ${\rm P}_{\varphi_{x}}$ measurable subsets of $\Theta$. A little calculation shows that the quantum posterior risk
    \begin{align}
        {\rm R}_{\varphi}(\delta)&={\rm E}_{\theta\lvert x}{\rm E}_{y}L(\theta,y)\\
        &=\sum_{i=1}^{n}\delta_{i\delta(x)}{\rm P}_{\varphi_{x}}(\theta\notin\Theta_{i})
    \end{align}
    and apparently ${\rm R}_{\varphi}(\delta)$ takes its minimum when $\delta(x)=\mathop{\arg\max}_{i\in[n]}{\rm P}_{\varphi_{x}}(\theta\in\Theta_{i})$ so that the quantum posterior solution $\delta^{*}$ w.r.t. $(\varphi,\mathcal{D})$ is
\begin{align}
    \delta^{*}(x,B)=1_{B}[\mathop{\arg\max}_{i\in[n]}{\rm P}_{\varphi_{x}}(\theta\in\Theta_{i})].
\end{align}
We say that the $\delta^{*}$ above is the quantum posterior testing rule of the multiple hypothesis testing
    \begin{align}
        H_{i}:\theta\in\Theta_{i},\ i\in[n].
    \end{align}
    
In classical Bayesian inference, the posterior predictive distribution of a $Z$-valued random variable is given by the Markov kernel $\int_{\Theta}K(\theta,x,C)\pi(\theta|x)\nu(d\theta)$, where $(Z,\mathscr{C})$ is a measurable space, $K:(\theta,x,C)\mapsto K(\theta,x,C),\ \forall (\theta,x,C)\in\Theta\times X\times\mathscr{C}$ a Markov kernel and $\pi(\theta|x)$ the posterior density. Note that the section $K_{x}$ of $K$ is not only an instrument but an observable as well. Let $\eta:\mathscr{C}\to\mathscr{M}$ be an observable. We say that $\langle\varphi_{x},\eta\rangle$ is the quantum posterior predictive distribution of $\eta$.

\section{Discussion}\label{sec6}
In a sense, quantum Bayes' rule tells a quantum Bayesian how to update her knowledge of an object based on observations. According to quantum Bayes' rule, a quantum analogue of Bayesian inference, which not only retains the classical one as a special case but possesses many new features as well, is put forward. However, there is still a lot to explore. 

(\romannumeral1) What is the sufficient and necessary conditions that a quantum posterior solution w.r.t. $(\varphi,\Delta)$ is a quantum Bayes solution w.r.t. $(\Pi,\Delta)$ and vice versa? 

(\romannumeral2) In classical Bayesian inference, if the true value of the unknown parameter $\theta$ is $\theta_{0}$, then
    \begin{align}
    \Pi_{0,x}=\cfrac{\int 1_{(\cdot)}p(x|\theta)\Pi_{0}(d\theta)}{\int_{\Theta} p(x|\theta)\Pi_{0}(d\theta)}=\Pi_{0},
    \end{align}
where $p(x\lvert\theta)$ is a nonnegative $\mathscr{E}\times\mathscr{A}$ measurable real valued function satisfying $\int_{X}p(x\lvert\theta)\mu(dx)=1,\ \forall\theta\in\Theta$, $\mu$ is a $\sigma$-finite measure on $(X,\mathscr{A})$. This indicates that $\Pi_{0}$ is never disturbed by a measurement with an apparatus corresponding to an instrument $K=\int 1_{(\cdot)}p(x|\theta)\mu(dx)$. However, in the definition of the weak consistency of posterior normal state, the ground truth $\varphi_{0}$ is likely to be disturbed by a measurement with an apparatus corresponding to an instrument $\mathfrak{I}_{n}$. So is it necessary for $\varphi_{0}$ to satisfy for a sequence of instruments $\left\{\mathfrak{I}_{n}\in\text{Ins}\left(\prod_{i=1}^{n}\mathscr{A}_{i},\mathscr{M}_{*}\right)\right\}_{n=1}^{\infty}$ such that $\langle\mathfrak{I}_{n+1}(A^{(n)}\times X_{n+1})\varphi,1\rangle=\langle\mathfrak{I}_{n}(A^{(n)})\varphi,1\rangle$ for any $A^{(n)}\in\prod_{i=1}^{n}\mathscr{A}_{i}$ and $n\in\mathbb{N}_{+}$ we have $\varphi_{0,X^{(n)}}=\varphi_{0}\ (\text{a.s.}),n\in\mathbb{N}_{+}$ and  $\langle\varphi_{0},\lambda\rangle=\Pi_{0}$?

(\romannumeral3) What is the sufficient and necessary conditions for the weak consistency of  posterior normal state?

(\romannumeral4) Topics concerning covariant quantum posterior point estimator.

(\romannumeral5) Choice and robustness of prior normal state, e.g. global or local sensitivity measures of a class of prior normal states.

\section*{Acknowledgments}
The author would like to thank Professor Naihui Chen and Professor Weihua Liu for their discussion and advice on this paper.

\nocite{*}

\bibliography{apssamp}

\providecommand{\noopsort}[1]{}\providecommand{\singleletter}[1]{#1}%
\begin{thebibliography}{30}%
\makeatletter
\providecommand \@ifxundefined [1]{%
 \@ifx{#1\undefined}
}%
\providecommand \@ifnum [1]{%
 \ifnum #1\expandafter \@firstoftwo
 \else \expandafter \@secondoftwo
 \fi
}%
\providecommand \@ifx [1]{%
 \ifx #1\expandafter \@firstoftwo
 \else \expandafter \@secondoftwo
 \fi
}%
\providecommand \natexlab [1]{#1}%
\providecommand \enquote  [1]{``#1''}%
\providecommand \bibnamefont  [1]{#1}%
\providecommand \bibfnamefont [1]{#1}%
\providecommand \citenamefont [1]{#1}%
\providecommand \href@noop [0]{\@secondoftwo}%
\providecommand \href [0]{\begingroup \@sanitize@url \@href}%
\providecommand \@href[1]{\@@startlink{#1}\@@href}%
\providecommand \@@href[1]{\endgroup#1\@@endlink}%
\providecommand \@sanitize@url [0]{\catcode `\\12\catcode `\$12\catcode
  `\&12\catcode `\#12\catcode `\^12\catcode `\_12\catcode `\%12\relax}%
\providecommand \@@startlink[1]{}%
\providecommand \@@endlink[0]{}%
\providecommand \url  [0]{\begingroup\@sanitize@url \@url }%
\providecommand \@url [1]{\endgroup\@href {#1}{\urlprefix }}%
\providecommand \urlprefix  [0]{URL }%
\providecommand \Eprint [0]{\href }%
\providecommand \doibase [0]{https://doi.org/}%
\providecommand \selectlanguage [0]{\@gobble}%
\providecommand \bibinfo  [0]{\@secondoftwo}%
\providecommand \bibfield  [0]{\@secondoftwo}%
\providecommand \translation [1]{[#1]}%
\providecommand \BibitemOpen [0]{}%
\providecommand \bibitemStop [0]{}%
\providecommand \bibitemNoStop [0]{.\EOS\space}%
\providecommand \EOS [0]{\spacefactor3000\relax}%
\providecommand \BibitemShut  [1]{\csname bibitem#1\endcsname}%
\let\auto@bib@innerbib\@empty
\bibitem [{\citenamefont {Barndorff-Nielsen}, \citenamefont {Gill},\ and\
  \citenamefont {Jupp}(2003)}]{bib1}%
  \BibitemOpen
  \bibfield  {author} {\bibinfo {author} {\bibfnamefont {O.~E.}\ \bibnamefont
  {Barndorff-Nielsen}}, \bibinfo {author} {\bibfnamefont {R.~D.}\ \bibnamefont
  {Gill}},\ and\ \bibinfo {author} {\bibfnamefont {P.~E.}\ \bibnamefont
  {Jupp}},\ }\bibfield  {title} {\enquote {\bibinfo {title} {On quantum
  statistical inference},}\ }\href@noop {} {\bibfield  {journal} {\bibinfo
  {journal} {Journal of the Royal Statistical Society: Series B (Statistical
  Methodology)}\ }\textbf {\bibinfo {volume} {65}},\ \bibinfo {pages}
  {775--804} (\bibinfo {year} {2003})}\BibitemShut {NoStop}%
\bibitem [{\citenamefont {Gill}\ and\ \citenamefont
  {Gu{\c{t}}{\u{a}}}(2013)}]{bib2}%
  \BibitemOpen
  \bibfield  {author} {\bibinfo {author} {\bibfnamefont {R.~D.}\ \bibnamefont
  {Gill}}\ and\ \bibinfo {author} {\bibfnamefont {M.~I.}\ \bibnamefont
  {Gu{\c{t}}{\u{a}}}},\ }\enquote {\bibinfo {title} {On asymptotic quantum
  statistical inference},}\ \ (\bibinfo {year} {2013})\ p.\ \bibinfo {pages}
  {105–127}\BibitemShut {NoStop}%
\bibitem [{\citenamefont {Brody}\ and\ \citenamefont {Meister}(1996)}]{bib3}%
  \BibitemOpen
  \bibfield  {author} {\bibinfo {author} {\bibfnamefont {D.~C.}\ \bibnamefont
  {Brody}}\ and\ \bibinfo {author} {\bibfnamefont {B.}~\bibnamefont
  {Meister}},\ }\bibfield  {title} {\enquote {\bibinfo {title} {Bayesian
  inference in quantum systems},}\ }\href@noop {} {\bibfield  {journal}
  {\bibinfo  {journal} {Physica A: Statistical Mechanics and its Applications}\
  }\textbf {\bibinfo {volume} {223}},\ \bibinfo {pages} {348--364} (\bibinfo
  {year} {1996})}\BibitemShut {NoStop}%
\bibitem [{\citenamefont {Ban}(2016)}]{bib4}%
  \BibitemOpen
  \bibfield  {author} {\bibinfo {author} {\bibfnamefont {M.}~\bibnamefont
  {Ban}},\ }\bibfield  {title} {\enquote {\bibinfo {title} {Bayes cost of
  parameter estimation for a quantum system interacting with an environment},}\
  }\href@noop {} {\bibfield  {journal} {\bibinfo  {journal} {Quantum
  Information Processing}\ }\textbf {\bibinfo {volume} {15}},\ \bibinfo {pages}
  {2213--2230} (\bibinfo {year} {2016})}\BibitemShut {NoStop}%
\bibitem [{\citenamefont {Teo}, \citenamefont {Oh},\ and\ \citenamefont
  {Jeong}(2018)}]{bib5}%
  \BibitemOpen
  \bibfield  {author} {\bibinfo {author} {\bibfnamefont {Y.~S.}\ \bibnamefont
  {Teo}}, \bibinfo {author} {\bibfnamefont {C.}~\bibnamefont {Oh}},\ and\
  \bibinfo {author} {\bibfnamefont {H.}~\bibnamefont {Jeong}},\ }\bibfield
  {title} {\enquote {\bibinfo {title} {Bayesian error regions in quantum
  estimation i: analytical reasonings},}\ }\href@noop {} {\bibfield  {journal}
  {\bibinfo  {journal} {New Journal of Physics}\ }\textbf {\bibinfo {volume}
  {20}},\ \bibinfo {pages} {093009} (\bibinfo {year} {2018})}\BibitemShut
  {NoStop}%
\bibitem [{\citenamefont {Oh}, \citenamefont {Teo},\ and\ \citenamefont
  {Jeong}(2018)}]{bib6}%
  \BibitemOpen
  \bibfield  {author} {\bibinfo {author} {\bibfnamefont {C.}~\bibnamefont
  {Oh}}, \bibinfo {author} {\bibfnamefont {Y.~S.}\ \bibnamefont {Teo}},\ and\
  \bibinfo {author} {\bibfnamefont {H.}~\bibnamefont {Jeong}},\ }\bibfield
  {title} {\enquote {\bibinfo {title} {Bayesian error regions in quantum
  estimation ii: region accuracy and adaptive methods},}\ }\href@noop {}
  {\bibfield  {journal} {\bibinfo  {journal} {New Journal of Physics}\ }\textbf
  {\bibinfo {volume} {20}},\ \bibinfo {pages} {093010} (\bibinfo {year}
  {2018})}\BibitemShut {NoStop}%
\bibitem [{\citenamefont {Quadeer}, \citenamefont {Tomamichel},\ and\
  \citenamefont {Ferrie}(2019)}]{bib7}%
  \BibitemOpen
  \bibfield  {author} {\bibinfo {author} {\bibfnamefont {M.}~\bibnamefont
  {Quadeer}}, \bibinfo {author} {\bibfnamefont {M.}~\bibnamefont
  {Tomamichel}},\ and\ \bibinfo {author} {\bibfnamefont {C.}~\bibnamefont
  {Ferrie}},\ }\bibfield  {title} {\enquote {\bibinfo {title} {Minimax quantum
  state estimation under bregman divergence},}\ }\href@noop {} {\bibfield
  {journal} {\bibinfo  {journal} {Quantum}\ }\textbf {\bibinfo {volume} {3}},\
  \bibinfo {pages} {126} (\bibinfo {year} {2019})}\BibitemShut {NoStop}%
\bibitem [{\citenamefont {Lukens}\ \emph {et~al.}(2020)\citenamefont {Lukens},
  \citenamefont {Law}, \citenamefont {Jasra},\ and\ \citenamefont
  {Lougovski}}]{bib8}%
  \BibitemOpen
  \bibfield  {author} {\bibinfo {author} {\bibfnamefont {J.~M.}\ \bibnamefont
  {Lukens}}, \bibinfo {author} {\bibfnamefont {K.~J.}\ \bibnamefont {Law}},
  \bibinfo {author} {\bibfnamefont {A.}~\bibnamefont {Jasra}},\ and\ \bibinfo
  {author} {\bibfnamefont {P.}~\bibnamefont {Lougovski}},\ }\bibfield  {title}
  {\enquote {\bibinfo {title} {A practical and efficient approach for bayesian
  quantum state estimation},}\ }\href@noop {} {\bibfield  {journal} {\bibinfo
  {journal} {New Journal of Physics}\ }\textbf {\bibinfo {volume} {22}},\
  \bibinfo {pages} {063038} (\bibinfo {year} {2020})}\BibitemShut {NoStop}%
\bibitem [{\citenamefont {Ozawa}\ and\ \citenamefont
  {Khrennikov}(2020)}]{bib33}%
  \BibitemOpen
  \bibfield  {author} {\bibinfo {author} {\bibfnamefont {M.}~\bibnamefont
  {Ozawa}}\ and\ \bibinfo {author} {\bibfnamefont {A.}~\bibnamefont
  {Khrennikov}},\ }\bibfield  {title} {\enquote {\bibinfo {title} {Application
  of theory of quantum instruments to psychology: combination of question order
  effect with response replicability effect},}\ }\href@noop {} {\bibfield
  {journal} {\bibinfo  {journal} {Entropy}\ }\textbf {\bibinfo {volume} {22}}
  (\bibinfo {year} {2020})}\BibitemShut {NoStop}%
\bibitem [{\citenamefont {Schack}, \citenamefont {Brun},\ and\ \citenamefont
  {Caves}(2001)}]{bib9}%
  \BibitemOpen
  \bibfield  {author} {\bibinfo {author} {\bibfnamefont {R.}~\bibnamefont
  {Schack}}, \bibinfo {author} {\bibfnamefont {T.~A.}\ \bibnamefont {Brun}},\
  and\ \bibinfo {author} {\bibfnamefont {C.~M.}\ \bibnamefont {Caves}},\
  }\bibfield  {title} {\enquote {\bibinfo {title} {Quantum bayes rule},}\
  }\href@noop {} {\bibfield  {journal} {\bibinfo  {journal} {Physical Review
  A}\ }\textbf {\bibinfo {volume} {64}},\ \bibinfo {pages} {014305} (\bibinfo
  {year} {2001})}\BibitemShut {NoStop}%
\bibitem [{\citenamefont {Warmuth}(2005)}]{bib10}%
  \BibitemOpen
  \bibfield  {author} {\bibinfo {author} {\bibfnamefont {M.~K.}\ \bibnamefont
  {Warmuth}},\ }\bibfield  {title} {\enquote {\bibinfo {title} {A bayes rule
  for density matrices},}\ }\href@noop {} {\bibfield  {journal} {\bibinfo
  {journal} {Advances in Neural Information Processing Systems}\ }\textbf
  {\bibinfo {volume} {18}} (\bibinfo {year} {2005})}\BibitemShut {NoStop}%
\bibitem [{\citenamefont {Parzygnat}\ and\ \citenamefont
  {Russo}(2022)}]{bib12}%
  \BibitemOpen
  \bibfield  {author} {\bibinfo {author} {\bibfnamefont {A.~J.}\ \bibnamefont
  {Parzygnat}}\ and\ \bibinfo {author} {\bibfnamefont {B.~P.}\ \bibnamefont
  {Russo}},\ }\bibfield  {title} {\enquote {\bibinfo {title} {A non-commutative
  bayes' theorem},}\ }\href@noop {} {\bibfield  {journal} {\bibinfo  {journal}
  {Linear Algebra and its Applications}\ }\textbf {\bibinfo {volume} {644}},\
  \bibinfo {pages} {28--94} (\bibinfo {year} {2022})}\BibitemShut {NoStop}%
\bibitem [{\citenamefont {Farenick}\ and\ \citenamefont
  {Kozdron}(2012)}]{bib13}%
  \BibitemOpen
  \bibfield  {author} {\bibinfo {author} {\bibfnamefont {D.}~\bibnamefont
  {Farenick}}\ and\ \bibinfo {author} {\bibfnamefont {M.~J.}\ \bibnamefont
  {Kozdron}},\ }\bibfield  {title} {\enquote {\bibinfo {title} {Conditional
  expectation and bayes' rule for quantum random variables and positive
  operator valued measures},}\ }\href@noop {} {\bibfield  {journal} {\bibinfo
  {journal} {Journal of Mathematical Physics}\ }\textbf {\bibinfo {volume}
  {53}},\ \bibinfo {pages} {042201} (\bibinfo {year} {2012})}\BibitemShut
  {NoStop}%
\bibitem [{\citenamefont {Coecke}\ and\ \citenamefont
  {Spekkens}(2012)}]{bib14}%
  \BibitemOpen
  \bibfield  {author} {\bibinfo {author} {\bibfnamefont {B.}~\bibnamefont
  {Coecke}}\ and\ \bibinfo {author} {\bibfnamefont {R.~W.}\ \bibnamefont
  {Spekkens}},\ }\bibfield  {title} {\enquote {\bibinfo {title} {Picturing
  classical and quantum bayesian inference},}\ }\href@noop {} {\bibfield
  {journal} {\bibinfo  {journal} {Synthese}\ }\textbf {\bibinfo {volume}
  {186}},\ \bibinfo {pages} {651--696} (\bibinfo {year} {2012})}\BibitemShut
  {NoStop}%
\bibitem [{\citenamefont {Vanslette}(2018)}]{bib15}%
  \BibitemOpen
  \bibfield  {author} {\bibinfo {author} {\bibfnamefont {K.}~\bibnamefont
  {Vanslette}},\ }\bibfield  {title} {\enquote {\bibinfo {title} {The quantum
  bayes rule and generalizations from the quantum maximum entropy method},}\
  }\href@noop {} {\bibfield  {journal} {\bibinfo  {journal} {Journal of Physics
  Communications}\ }\textbf {\bibinfo {volume} {2}},\ \bibinfo {pages} {025017}
  (\bibinfo {year} {2018})}\BibitemShut {NoStop}%
\bibitem [{\citenamefont {Ozawa}(2004)}]{bib27}%
  \BibitemOpen
  \bibfield  {author} {\bibinfo {author} {\bibfnamefont {M.}~\bibnamefont
  {Ozawa}},\ }\bibfield  {title} {\enquote {\bibinfo {title} {Uncertainty
  relations for noise and disturbance in generalized quantum measurements},}\
  }\href@noop {} {\bibfield  {journal} {\bibinfo  {journal} {Annals of
  Physics}\ }\textbf {\bibinfo {volume} {311}},\ \bibinfo {pages} {350--416}
  (\bibinfo {year} {2004})}\BibitemShut {NoStop}%
\bibitem [{\citenamefont {Ozawa}(1984)}]{bib25}%
  \BibitemOpen
  \bibfield  {author} {\bibinfo {author} {\bibfnamefont {M.}~\bibnamefont
  {Ozawa}},\ }\bibfield  {title} {\enquote {\bibinfo {title} {Quantum measuring
  processes of continuous observables},}\ }\href@noop {} {\bibfield  {journal}
  {\bibinfo  {journal} {Journal of Mathematical Physics}\ }\textbf {\bibinfo
  {volume} {25}},\ \bibinfo {pages} {79--87} (\bibinfo {year}
  {1984})}\BibitemShut {NoStop}%
\bibitem [{\citenamefont {Okamura}\ and\ \citenamefont {Ozawa}(2016)}]{bib18}%
  \BibitemOpen
  \bibfield  {author} {\bibinfo {author} {\bibfnamefont {K.}~\bibnamefont
  {Okamura}}\ and\ \bibinfo {author} {\bibfnamefont {M.}~\bibnamefont
  {Ozawa}},\ }\bibfield  {title} {\enquote {\bibinfo {title} {Measurement
  theory in local quantum physics},}\ }\href@noop {} {\bibfield  {journal}
  {\bibinfo  {journal} {Journal of Mathematical Physics}\ }\textbf {\bibinfo
  {volume} {57}},\ \bibinfo {pages} {015209} (\bibinfo {year}
  {2016})}\BibitemShut {NoStop}%
\bibitem [{\citenamefont {Ozawa}(1985{\natexlab{a}})}]{bib20}%
  \BibitemOpen
  \bibfield  {author} {\bibinfo {author} {\bibfnamefont {M.}~\bibnamefont
  {Ozawa}},\ }\bibfield  {title} {\enquote {\bibinfo {title} {Conditional
  probability and a posteriori states in quantum mechanics},}\ }\href@noop {}
  {\bibfield  {journal} {\bibinfo  {journal} {Publications of the Research
  Institute for Mathematical Sciences}\ }\textbf {\bibinfo {volume} {21}},\
  \bibinfo {pages} {279--295} (\bibinfo {year}
  {1985}{\natexlab{a}})}\BibitemShut {NoStop}%
\bibitem [{\citenamefont {Davies}\ and\ \citenamefont {Lewis}(1970)}]{bib24}%
  \BibitemOpen
  \bibfield  {author} {\bibinfo {author} {\bibfnamefont {E.~B.}\ \bibnamefont
  {Davies}}\ and\ \bibinfo {author} {\bibfnamefont {J.~T.}\ \bibnamefont
  {Lewis}},\ }\bibfield  {title} {\enquote {\bibinfo {title} {An operational
  approach to quantum probability},}\ }\href@noop {} {\bibfield  {journal}
  {\bibinfo  {journal} {Communications in Mathematical Physics}\ }\textbf
  {\bibinfo {volume} {17}},\ \bibinfo {pages} {239--260} (\bibinfo {year}
  {1970})}\BibitemShut {NoStop}%
\bibitem [{\citenamefont {Bruneau}, \citenamefont {Joye},\ and\ \citenamefont
  {Merkli}(2014)}]{bib23}%
  \BibitemOpen
  \bibfield  {author} {\bibinfo {author} {\bibfnamefont {L.}~\bibnamefont
  {Bruneau}}, \bibinfo {author} {\bibfnamefont {A.}~\bibnamefont {Joye}},\ and\
  \bibinfo {author} {\bibfnamefont {M.}~\bibnamefont {Merkli}},\ }\bibfield
  {title} {\enquote {\bibinfo {title} {Repeated interactions in open quantum
  systems},}\ }\href@noop {} {\bibfield  {journal} {\bibinfo  {journal}
  {Journal of Mathematical Physics}\ }\textbf {\bibinfo {volume} {55}},\
  \bibinfo {pages} {075204} (\bibinfo {year} {2014})}\BibitemShut {NoStop}%
\bibitem [{\citenamefont {Movassagh}\ and\ \citenamefont
  {Schenker}(2022)}]{bib22}%
  \BibitemOpen
  \bibfield  {author} {\bibinfo {author} {\bibfnamefont {R.}~\bibnamefont
  {Movassagh}}\ and\ \bibinfo {author} {\bibfnamefont {J.}~\bibnamefont
  {Schenker}},\ }\bibfield  {title} {\enquote {\bibinfo {title} {An ergodic
  theorem for quantum processes with applications to matrix product states},}\
  }\href@noop {} {\bibfield  {journal} {\bibinfo  {journal} {Communications in
  Mathematical Physics}\ } (\bibinfo {year} {2022})}\BibitemShut {NoStop}%
\bibitem [{\citenamefont {Lim}(2010)}]{bib21}%
  \BibitemOpen
  \bibfield  {author} {\bibinfo {author} {\bibfnamefont {B.~J.}\ \bibnamefont
  {Lim}},\ }\emph {\bibinfo {title} {Poisson boundaries of quantum operations
  and quantum trajectories}},\ \href
  {https://tel.archives-ouvertes.fr/tel-00637636} {\bibinfo {type} {Theses}},\
  \bibinfo  {school} {{Universit{\'e} Rennes 1}} (\bibinfo {year}
  {2010})\BibitemShut {NoStop}%
\bibitem [{\citenamefont {Busemeyer}\ and\ \citenamefont
  {Trueblood}(2009)}]{bib11}%
  \BibitemOpen
  \bibfield  {author} {\bibinfo {author} {\bibfnamefont {J.~R.}\ \bibnamefont
  {Busemeyer}}\ and\ \bibinfo {author} {\bibfnamefont {J.}~\bibnamefont
  {Trueblood}},\ }\bibfield  {title} {\enquote {\bibinfo {title} {Comparison of
  quantum and bayesian inference models},}\ }in\ \href@noop {} {\emph {\bibinfo
  {booktitle} {International Symposium on Quantum Interaction}}}\ (\bibinfo
  {organization} {Springer},\ \bibinfo {year} {2009})\ pp.\ \bibinfo {pages}
  {29--43}\BibitemShut {NoStop}%
\bibitem [{\citenamefont {McLaren}, \citenamefont {Plosker},\ and\
  \citenamefont {Ramsey}(2017)}]{bib16}%
  \BibitemOpen
  \bibfield  {author} {\bibinfo {author} {\bibfnamefont {D.}~\bibnamefont
  {McLaren}}, \bibinfo {author} {\bibfnamefont {S.}~\bibnamefont {Plosker}},\
  and\ \bibinfo {author} {\bibfnamefont {C.}~\bibnamefont {Ramsey}},\
  }\bibfield  {title} {\enquote {\bibinfo {title} {On operator valued
  measures},}\ }\href@noop {} {\bibfield  {journal} {\bibinfo  {journal}
  {Houston Journal of Mathematics}\ } (\bibinfo {year} {2017})}\BibitemShut
  {NoStop}%
\bibitem [{\citenamefont {Busch}\ \emph {et~al.}(2016)\citenamefont {Busch},
  \citenamefont {Lahti}, \citenamefont {Pellonp{\"a}{\"a}},\ and\ \citenamefont
  {Ylinen}}]{bib17}%
  \BibitemOpen
  \bibfield  {author} {\bibinfo {author} {\bibfnamefont {P.}~\bibnamefont
  {Busch}}, \bibinfo {author} {\bibfnamefont {P.}~\bibnamefont {Lahti}},
  \bibinfo {author} {\bibfnamefont {J.-P.}\ \bibnamefont {Pellonp{\"a}{\"a}}},\
  and\ \bibinfo {author} {\bibfnamefont {K.}~\bibnamefont {Ylinen}},\ }\enquote
  {\bibinfo {title} {Quantum measurement},}\ \ (\bibinfo {year}
  {2016})\BibitemShut {NoStop}%
\bibitem [{\citenamefont {Pellonp{\"a}{\"a}}(2012)}]{bib19}%
  \BibitemOpen
  \bibfield  {author} {\bibinfo {author} {\bibfnamefont {J.-P.}\ \bibnamefont
  {Pellonp{\"a}{\"a}}},\ }\bibfield  {title} {\enquote {\bibinfo {title}
  {Quantum instruments: I. extreme instruments},}\ }\href@noop {} {\bibfield
  {journal} {\bibinfo  {journal} {Journal of Physics A: Mathematical and
  Theoretical}\ }\textbf {\bibinfo {volume} {46}},\ \bibinfo {pages} {025302}
  (\bibinfo {year} {2012})}\BibitemShut {NoStop}%
\bibitem [{\citenamefont {Ozawa}(1985{\natexlab{b}})}]{bib26}%
  \BibitemOpen
  \bibfield  {author} {\bibinfo {author} {\bibfnamefont {M.}~\bibnamefont
  {Ozawa}},\ }\bibfield  {title} {\enquote {\bibinfo {title} {Concepts of
  conditional expectations in quantum theory},}\ }\href@noop {} {\bibfield
  {journal} {\bibinfo  {journal} {Journal of Mathematical Physics}\ }\textbf
  {\bibinfo {volume} {26}},\ \bibinfo {pages} {1948--1955} (\bibinfo {year}
  {1985}{\natexlab{b}})}\BibitemShut {NoStop}%
\bibitem [{\citenamefont {Carmeli}, \citenamefont {Heinosaari},\ and\
  \citenamefont {Toigo}(2011)}]{bib28}%
  \BibitemOpen
  \bibfield  {author} {\bibinfo {author} {\bibfnamefont {C.}~\bibnamefont
  {Carmeli}}, \bibinfo {author} {\bibfnamefont {T.}~\bibnamefont
  {Heinosaari}},\ and\ \bibinfo {author} {\bibfnamefont {A.}~\bibnamefont
  {Toigo}},\ }\bibfield  {title} {\enquote {\bibinfo {title} {Sequential
  measurements of conjugate observables},}\ }\href@noop {} {\bibfield
  {journal} {\bibinfo  {journal} {Journal of Physics A: Mathematical and
  Theoretical}\ }\textbf {\bibinfo {volume} {44}},\ \bibinfo {pages} {285304}
  (\bibinfo {year} {2011})}\BibitemShut {NoStop}%
\bibitem [{\citenamefont {Berger}(1985)}]{bib29}%
  \BibitemOpen
  \bibfield  {author} {\bibinfo {author} {\bibfnamefont {J.}~\bibnamefont
  {Berger}},\ }\href {https://doi.org/10.1007/978-1-4757-1727-3} {\emph
  {\bibinfo {title} {Statistical Decision Theory and Bayesian Analysis}}}\
  (\bibinfo {year} {1985})\BibitemShut {NoStop}%
\end{thebibliography}%

\end{document}